\documentclass[authoryear]{elsarticle}
	
	\usepackage{amsmath,amssymb,amsfonts,amsthm}
	\usepackage{natbib}
	\usepackage{graphicx}
	\usepackage{endfloat}
\newtheorem{defi}{Definition}[section]
\newtheorem{thm}[defi]{Theorem}

\newtheorem{cor}[defi]{Corollary}

\newcommand{\by}[2]{#1 \times #2}

\newcommand{\trans}{^{\scriptscriptstyle T}}    

\newcommand{\var}{\mathrm{Var}}         
\newcommand{\cov}{\mathrm{Cov}}         

\newcommand{\ti}{t}                         
\newcommand{\comp}{c}                       
\newcommand{\loc}{l}                        
\newcommand{\lc}{r}                         
\newcommand{\X}{X}                          
\newcommand{\Y}{Y}                          
\newcommand{\CR}{\alpha}                    
\newcommand{\epix}[1]{\epsilon_{\X #1}}     
\newcommand{\epiy}[1]{\epsilon_{\Y #1}}     
\newcommand{\epithet}[1]{\epsilon_{\theta #1}}     
\newcommand{\epialp}[1]{\epsilon_{\CR #1}}  
\newcommand{\epilc}[1]{\epsilon_{\lc #1}}   
\newcommand{\epixn}{\epsilon_{\X}}          
\newcommand{\epialpn}{\epsilon_{\CR}}       
\newcommand{\E}{\xi}                          
\newcommand{\Sigy}{\Sigma_{\Y}}             
\newcommand{\thet}{\Theta}                  
\newcommand{\In}{1}                         
\newcommand{\On}{0}                         

\newcommand{\M}[1]{\mathcal{M}(#1)}         
\newcommand{\Res}[2]{\mathcal{R}_{#1}(#2)}  

\newcommand{\musigX}{\mu_{W_\X}}           
\newcommand{\varsigX}{\Sigma_{W_\X}}       
\newcommand{\musigthet}{\mu_{W_\theta}}           
\newcommand{\varsigthet}{\Sigma_{W_\theta}}       
\newcommand{\covsigthet}{\Gamma_{W_\theta}}       
\newcommand{\covsigX}{\Gamma_{W_\X}}       

\newcommand{\musigCR}{\mu_{W_\CR}}           
\newcommand{\musiglc}{\Sigma_\lc}           

\newcommand{\Dbar}{\bar{\Delta}}                 

	\title{Bayes linear variance structure learning for inspection of large scale physical systems}
	
	\author[dur]{D. Randell}
	\author[dur]{M. Goldstein}
	\author[thorn]{P. Jonathan}
	\address[dur]{University of Durham, UK}
	\address[thorn]{Shell Projects and Technology, Chester, UK}
	
	\journal{Journal of Statistical Planning and Inference}
	
\begin{document}
	
	\begin{frontmatter}
	
	\begin{abstract}
	Modelling of inspection data for large scale physical systems is critical to assessment of their integrity. We present a general method for inference about system state and associated model variance structure from spatially distributed time series which are typically short, irregular, incomplete and not directly observable. Bayes linear analysis simplifies parameter estimation and avoids often-unrealistic distributional assumptions. Second-order exchangeability judgements facilitate variance learning for sparse inspection time-series. The model is applied to inspection data for minimum wall thickness from corroding pipe-work networks on a full-scale offshore platform, and shown to give materially different forecasts of remnant life compared to an equivalent model neglecting variance learning.
	\end{abstract}
	
	\begin{keyword}
	Bayes linear, exchangeability, variance learning, corrosion, dynamic linear model, Mahalanobis distance
	\end{keyword}
	
	\end{frontmatter}
	
	\section{Introduction}
	
	Inspection and maintenance assures the integrity of physical systems subject to degradation (e.g. to corrosion and fouling in time). Inspection and maintenance is typically costly, requiring careful allocation of limited resources. Statistical modelling provides one approach to ensuring the effectiveness of inspection and maintenance activities.
	
	Inference for complete systems consisting of thousands of components is a largely overlooked aspect of real world inspection planning and inference, due to methodological and computational complexity. Most attempts to model a degrading system empirically consider individual system components. Yet characteristics of degradation are often common across components, due to component design, age, location, manner of operation, etc. A multi-component (or system-wide) model could make use of common between-component behaviour to improve the quality of inspection information and hence achieve more efficient inspection planning. Inspections at a given time are rarely performed system-wide. A realistic empirical model for a large scale system will therefore assume partial inspection data as input.
	
	Here we develop a model for degradation of a system consisting of multiple dependent components, applicable to the analysis of irregularly spaced spatially distributed short time series. We make inferences about system state per component and model variances, given indirect observations. A Bayes linear approach simplifies parameter estimation in comparison with a full Bayesian analysis and avoids unrealistic distributional assumptions, using model-based simulation. Diagnostic tests assess model fit, and simulation studies based on a hypothetical known system are used to evaluate performance. We illustrate the method by modelling wall thickness and corrosion rates for corroding pipe-work networks on a full-scale offshore platform, given sparse inspection data for component minimum wall thickness.
	
Specification of realistic priors and initial values (e.g. \cite{oha06}), for model variances in particular (e.g. \cite{far03}), is problematic in general yet can be highly influential, particularly when data are sparse. To avoid over-reliance on poorly-specified priors and initial values, we develop and implement methods of inference for the model variance structure. We facilitate variance estimation by making exchangeability judgements appropriate for analysis for sparse time-series of irregular partial inspections.
	
	Bayes linear methods (\citet{BLS07}) applied to dynamic linear models (DLM, \citet{HW97}) offer system-wide modelling of corroding systems using partial inspection data. \citet{little1} uses a multivariate DLM to characterise the corrosion of large industrial storage tanks, using observations of component minima, and suggests approaches to optimal inspection planning. \citet{little2} describes the application of a spatio-temporal DLM to model the corrosion of an industrial furnace using Bayes linear updating. Empirical distance-based estimates for covariances of DLM observation and system variances are used, and optimal inspection planning based on heuristic criteria is considered. \cite{Far93} discusses Bayes linear methods for grouped multivariate repeated measurement studies with application to cross-over trials.  \cite{wilk97a} discusses variance learning for a univariate linear growth DLM, and \cite{wilk97b} describes Bayes linear covariance matrix adjustment for a multivariate constant DLM. \cite{sha99} discusses Bayes linear experimental design for grouped multivariate exchangeable systems. \cite{ran10} develops a utility-based criterion to assess inspection design quality, based on a linear growth DLM with variance learning.
	Industry guidelines (e.g. \citet{HS02} and \citet{ASTM}) treat the modelling of corrosion very generally, yet there is a vast body of engineering literature on this subject. \citet{zhang00} outlines mathematical expressions for initiation and evolution of different corrosion mechanisms, including pitting and cracking. \citet{kall00} discusses inspection and maintenance decisions based on imperfect inspection within a Bayesian framework, using gamma processes.  \citet{qin03} considers corrosion of steel structures, and \cite{gase01} presents a Bayesian approach using partial inspections only. A number of authors discuss the inclusion of inspection data and expert judgement within a risk-based inspection framework. For example, \citet{farb00c} presents an approach to estimating system condition for inspection planning purposes using a combination of inspection observations and expert judgement, and \cite{strb04} describes generic approaches to risk-based inspection of steel structures. \cite{kun09} presents a method for compliance sampling. Many approaches to corrosion modelling make use of methods associated with extreme value analysis (e.g.  \cite{col01}); \cite{scr96} provides a review. \cite{gle07} applies the generalised extreme value distribution (e.g. \cite{ktz00}), and \cite{lop08} the non-homogeneous Poisson process (e.g. \cite{dal03}).

A number of aspects of the current work are novel, including Bayes linear variance learning using partial observations and Mahalanobis learning for local corrosion variances. The statistical model adopted admits a non-linear observation equation. We also believe that the application of the methodology to a full-scale industrial system rather than individual components is particularly interesting and informative.

The article is presented as follows. We start by outlining the motivating application, a full-scale offshore platform, in section \ref{data}, and Bayes linear methods in section \ref{theory}. Section \ref{model} introduces the corrosion model and in particular discusses the exchangeability judgements made to accommodate irregular partial inspection data. Section \ref{BLinf} discusses Bayes linear inference, and the estimation of model variances, involving a Mahalanobis distance fitting procedure for local corrosion variance. Estimation and performance of the variance learning, in application to simulated data with the same inspection design and covariance structure as the historical inspection data, is also discussed. To facilitate future application of the method, section \ref{alg} presents a stepwise modelling procedure. The model is applied to historical data from the offshore platform in section \ref{examcorm}, and diagnostic tools for quality of model fit illustrated. Section \ref{conc} discusses findings and proposes possible future generalisations.  \ref{App:EDMW} outlines calculations used in computing the adjusted expectations of model variances. \ref{priorreal} provides prior values used in the analysis of both historical and simulated inspection data.
	
	\section{Motivating application} \label{data} \label{exampleintro}
	
	We consider inspection of a full-scale offshore platform as a motivating application. For inspection and maintenance purposes, the installation is considered as a set of corrosion circuits, each consisting of multiple components, for inspection. For the current application, we model a system of four corrosion circuits consisting of a total of $64$ pipe-work weld components. A corrosion circuit typically exhibits long sequences of connected components, often with side branches of components. The corrosion behaviour of components is influenced in part by their connectivity within a corrosion circuit (see section \ref{wallthickevo}). A corroding system may exhibit multiple corrosion circuits with similar corrosion behaviour. This application was discussed previously by \cite{ran10}.
	
When components are inspected, inspectors are most concerned with most vulnerable components since these are critical to assessing risk of system failure. Measurements are made using non-destructive (e.g. ultrasonic) inspection. The use of non-destructive testing means that operations can continue with little or no system interruption. The measurement device reports an estimate for the minimum wall thickness over a region sometimes referred to as the footprint of the device. Historical data for component minimum wall thickness, obtained during inspection campaigns for the period $1998$ - $2005$, are available. Based on the frequency of observations and the requirements for inspection planning, we select a monthly time increment for modelling. The sample period therefore consists of $83$ time points.
	
	The actual historical inspection design is given in figure \ref{siminspdes}. It is clear that inspections are typically incomplete and irregularly spaced in time. A total of 174 observations of the system are available.	The application is used in two ways in what follows. Firstly, it is used to evaluate model performance by generating simulated samples using the actual inspection scheme and realistic parameter values. Secondly, in section \ref{examcorm}, estimates for system parameters based on actual historical inspection data, and system forecasts, are produced. 	For both simulation and historical data analyses, the actual inspection design matrix is used, and identical prior specifications made, so that the analyses are comparable. 

Estimates for prior values are obtained from two sources, namely auxiliary historical data from other corrosion circuits for the same offshore installation, and judgements made by experienced inspection engineers. (A full list of prior values is given in section \ref{priorreal}.)

	\begin{figure}
	\includegraphics[width=\textwidth]{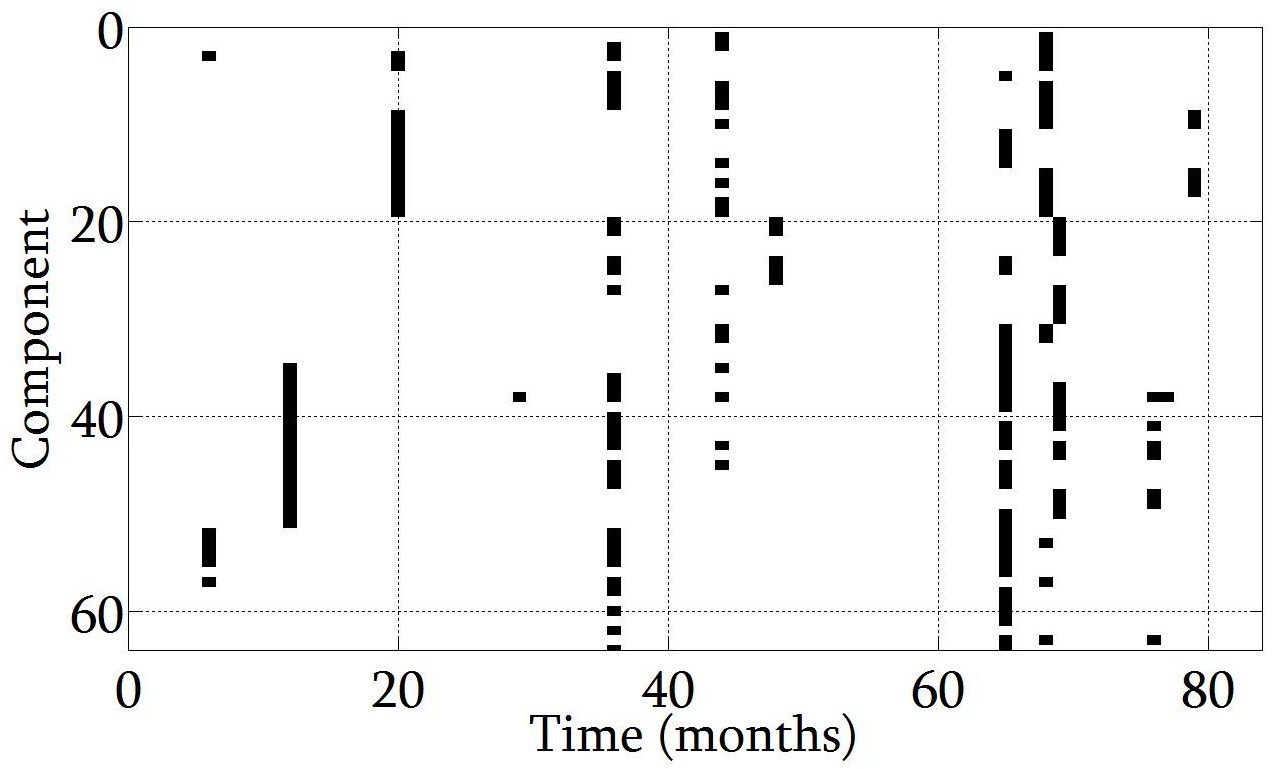}
	\caption{Inspection design for the offshore application, consisting of 64 components over 83 time points. Black lines correspond to 174 observations of the system.} \label{siminspdes}
	\end{figure}

	\section{Theoretical background} \label{theory}
	
	\subsection{Bayes linear analysis} \label{BL}
	
	For a complex system, it can be difficult or impractical to make full prior belief specifications. Bayes linear analysis allows us to specify and update partial aspects of our beliefs. Bayes linear analysis also provides a computationally efficient method for updating beliefs for problems where a full Bayes approach would be too difficult or time consuming. It can also be viewed as a generalisation of the full Bayes approach which relaxes the requirement for full probabilistic prior specifications. In Bayes linear analysis, expectation rather than probability is treated as a primitive quantity; prior beliefs are specified in terms of means, variances and covariances. A detailed explanation is given in \citet{BLS07}. Discussion of the application of these methods for the analysis of simulators for large systems is given in \cite{CompSim01} and \cite{CompSim06}.
	
Given a vector of data $D$, the adjusted expectation $E_D(B)$ for a vector $B$ is given by:
\begin{equation*}
 E_D(B)=E(B)+\cov(B,D)(\var(D))^{\dagger}(d-E(D))
\end{equation*}
where the matrix $\var(D)^{\dagger}$ is the inverse of $\var(D)$, if invertible, or a generalised inverse otherwise. The variance resolved by adjustment, $\mathrm{Rvar}_D(B)$, is given by:
\begin{equation} \label{Rvar}
\mathrm{Rvar}_D(B)=\cov(B,D)(\var(D))^{\dagger}\cov(D,B)
\end{equation}
and the adjusted variance, $\var_D(B)$, by:
\begin{equation*}
 \var_D(B)=\var(B)-\mathrm{Rvar}_D(B)
\end{equation*}

	\subsection{Exchangeability and the representation theorem } \label{exchang}
	
	The concept of exchangeable events is a crucial component of the subjective theory of probability. In essence, exchangeability judgements in a subjective analysis can be used to underpin the types of independence assumptions made in classical inference \citep{TP74}. For Bayes linear analysis, where only partial beliefs need to be specified, we can restrict our assumptions for the error structures to exchangeability of the first and second order quantities.
	
	The means, variance and covariances of a second order exchangeable sequence, $X={X_1,X_2,\dots}$ , are invariant under permutation. If we assume second-order exchangeability, we can use the the second order exchangeability representation theorem, \citep{EBS86} to express the quantities $X_i$ in the sequence in terms of the sum of two random quantities $\mathcal{M}(X)$ and $\mathcal{R}_i (X)$ which may be viewed as analogous to an underlying population mean and discrepancies from the mean respectively, as follows. Given a collection of vectors, $X={X_1,X_2,\dots}$, an infinitely second order exchangeable sequence with:
	\[ E(X_i)=\mu_X \text{, } \var(X_i)=\Sigma_X \text{ and } \cov(X_i,X_j)=\Gamma_X \quad i\neq j \]
we can express each $X_i$ as:
	\[X_i =\mathcal{M}(X)+\mathcal{R}_i (X) \]
	where $\mathcal{M}(X)$ is a random vector known as the population mean with:
	\begin{equation} \label{repthm1}
	E(\mathcal{M}(X))=\mu_X \quad \var (\mathcal{M}(X))=\Gamma_X
	\end{equation}
	and the discrepancies $\mathcal{R}_i (X)$, themselves second order exchangeable, with:
	\begin{equation} \label{repthm2}
	E(\mathcal{R}_i(X))=0 \text{ and } \var(\mathcal{R}_i(X))=\Sigma_X-\Gamma_X
	\end{equation}
	Each pair $\mathcal{R}_i$ and $\mathcal{R}_j$ is uncorrelated ($i \neq j$) and each $\mathcal{R}_i$ is uncorrelated
	with $\mathcal{M}(X)$.
	
	\subsection{Bayes linear inference and its analogy to ``full'' Bayes}
	
	Bayes linear inference can be understood by analogy to the usual ``full'' Bayes inference as follows. Suppose we observe a quantity $D$ and use it to make inferences about an unknown quantity $B$ expressed as $B= \mathcal{M} + \mathcal{R}$, where $\mathcal{M}$ and $\mathcal{R}$ are unknown random quantities with full prior distribution $f(\mathcal{M},\mathcal{R})$. The prior is related to the full posterior distribution $f(\mathcal{M}, \mathcal{R} | D)$ by Bayes theorem:
	\[  f(\mathcal{M}, \mathcal{R} | D) \propto f(D |\mathcal{M}, \mathcal{R})  f(\mathcal{M}, \mathcal{R}) \]
	where $f(D |\mathcal{M}, \mathcal{R}) $ is the likelihood for $D$. In the Bayes linear formalism, the second order exchangeability assumptions and associated representation theorem  $B= \mathcal{M} + \mathcal{R}$ provide a structure for $B$. Bayes linear adjustment is then analogous to estimation of the full posterior distribution for  $\mathcal{M}$ and  $\mathcal{R}$. The specific properties of  $\mathcal{M}$ and  $\mathcal{R}$ imposed allow inferences to be made relatively straightforwardly. The non random quantities $E( \mathcal{M})$ and $E( \mathcal{R})$, primitives in Bayes linear inference, are analogous to prior distributions. The non random adjusted expectations $E_D( \mathcal{M})$ and $E_D( \mathcal{M})$ are analogous to posterior distributions. A fuller account is given by \citet[section 3.5]{BLS07} 
	
	\section{Model} \label{model}
	
	\subsection{General framework} \label{genfrmwrk}
	
	The general framework for our analysis is as follows. We assume that the system to be inspected can be partitioned into a set of $C$ components, indexed by $\comp$, whose characteristics evolve over $T$ time points, indexed by $\ti$. We seek inferences about the true system state vector $Z_{\loc\comp\ti}$ over $L$ locations indexed by $\loc$ within each component. We separate global aspects which affect the whole component from local aspects. This allows us to distinguish between different model characteristics.
	
\subsubsection*{Global effects}
The global effects model captures the most important features of, and relationships between components. Global effects evolve in time as a dynamic linear model (DLM) with system evolution matrix $G$:
\begin{equation}\thet_\ti = G \thet_{\ti-1} +\epithet{\ti} \label{eq:lingrodlm}\end{equation}
where $\thet_\ti$ and $\epithet{\ti}$ are vectors over components with elements $\thet_{c\ti}$ and $\epithet{c\ti}$ respectively.

\subsubsection*{Local effects}
The local effects model describes spatial variability in detail. Local effects $\lc_{\loc\comp\ti}$ evolve in time for some function $g$ as:
\begin{equation}\lc_{\loc \comp \ti} = g(\lc_{\loc \comp (\ti-1)}) + \epilc{\loc \comp \ti}\label{eq:loccorrmod}\end{equation}
	
\subsubsection*{True System State}	
We model $Z_{\loc \comp\ti}$  as the sum of global and local effects:
\begin{equation}Z_{\loc \comp\ti} =F_{\comp} \thet_\ti + \lc_{\loc\comp \ti}\label{eq:trustate}\end{equation}
where $F_\comp$ is the $\comp$th row of the matrix $F$ of linear combinations of global effects parameter vector $\theta_{\ti}$.
	
\subsubsection*{Observations}
For each component and time, we can choose at a cost to observe a function $f$ of the true system state vector over locations $Z_{\comp\ti}=\left(\begin{array}{ccc} \dots & Z_{\loc\comp\ti} & \dots \end{array}\right) \trans$ with error:
\begin{equation}\Y_{\comp\ti} = f\left(  Z_{\comp \ti} + \epiy{\comp\ti} \right)\label{eq:obsproc}\end{equation}
where $\epiy{\comp\ti}$ is a vector of measurement errors across locations.  We assume that the function $f$ is non--linear and separable in the following sense. For vectors $a$ and $\boldsymbol{1}$ (a vectors ones) of the same length, and a scalar $b$, $f$ can be decomposed as:
\[f(a+b\boldsymbol{1})=f(a)+b\]
Thus:
\begin{align*}
f\left(  Z_{\comp \ti} + \epiy{\comp\ti} \right)&=f\left(  F \thet_\ti +\lc_{\loc \comp \ti} + \epiy{\comp\ti} \right)\\
&= F \thet_\ti+f\left( \lc_{\loc \comp \ti} + \epiy{\comp\ti} \right)
\end{align*}
We note that separability applies to many functions for summarising data, including the mean, the median, maximum, minimum and other quantiles.

\subsubsection*{Complete Model}
	 \begin{align*}
	\textrm{Observation Equation:}& &\Y_{\comp\ti} &= f\left(  Z_{\comp \ti} + \epiy{\comp\ti} \right)\\
		\textrm{True System State:}  &&Z_{\loc \comp\ti} &=F_{\comp} \thet_\ti + \lc_{\loc\comp \ti}\\
\textrm{Global Effects Model:}& & \thet_\ti &= G \thet_{\ti-1} +\epithet{\ti} \\
	\textrm{Local Effects Model:}& & \lc_{\loc \comp \ti} &= g(\lc_{\loc \comp (\ti-1)}) + \epilc{\loc \comp \ti}
	 \end{align*}
We assume that errors $\epithet{\ti}$, $\epilc{\loc\comp \ti}$ and $\epiy{\loc \comp \ti}$ are mutually uncorrelated in time. The elements $\epithet{\comp\ti}$ of $\epithet{\ti}$ are correlated across components with $\var(\epithet{\comp \ti})=\Sigma_\theta$, but the $\epilc{\loc \comp \ti}$ are independent across components. We assume $\var(\epilc{\loc\comp \ti})=\Sigma_\lc$ and  $\var(\epiy{\loc\comp \ti})=\Sigy$, both scalar constants. Specification of error structure in terms of second order exchangeability is discussed in section \ref{exchangassmp}.
	
\subsection{Exchangeability judgements}\label{exchangassmp}\label{wallthickevo}

Non-linearity of the observation equation makes formal Bayesian inference difficult. Coupled with the large system size and resulting difficulty of prior specification, full Bayesian calculations become intractable. This is particularly important when considering inspection design requiring fast evaluation since many samples are needed to assess each of a large number of design choices. Here, we proceed using Bayes linear methods (section \ref{BL}) in conjunction with plausible exchangeability judgements. The latter provide a means to define model variances from squared linear combinations of observations since they provide access to model evolution errors (section \ref{BLvar}). In order to invoke second order exchangeability we view the observations in time and space as part an infinite exchangeable sequence. We make the following exchangeability judgements:
	
	\subsubsection*{Judgement: Second order exchangeability of squared errors $\epithet{\comp\ti}^2$ over time for each component.} This leads to representation statements for the squared residuals of every component:
	\[\epithet{\comp \ti}^2=V_{\theta\comp \ti} = \M{V_{\theta\comp}}+\Res{\ti}{V_{\theta\comp}}\]
which can then be further decomposed.
	\subsubsection*{Judgement: Second order exchangeability of $\M{V_{\theta\comp}}$ across components
	such that:}%
	\[\M{V_{\theta\comp}}=W_{\theta\comp} =\M{W_{\theta}}+\Res{\comp}{W_{\theta}}\]
	where:
	\[E(W_{\theta\comp})=\musigthet \text{ , } \var(W_{\theta\comp})=\varsigthet \text{ and }
	\cov(W_{\theta\comp},W_{\theta\comp'})=\covsigthet \textrm{ for } \comp \neq \comp' \]
	so that:
	\[E(\epithet{\comp \ti}^2)=E(V_{\theta\comp \ti})=E(\M{V_{\theta\comp}})=E(W_{\theta\comp})=\musigthet \text{ and } \var(\M{V_{\theta\comp}})=\varsigthet\]
	Then, from equations \ref{repthm1} and \ref{repthm2} respectively:
	\[\var(\M{W_{\theta}})=\covsigthet \quad \textrm{and} \quad \var(\Res{\comp}{W_{\theta}})=\varsigthet-\covsigthet \]
	and for each $\epithet{\comp\ti}^2$ we can write:
	\begin{align}
	\epithet{\comp \ti}^2 &= \M{V_{\theta\comp}}+\Res{\ti}{V_{\theta\comp}} \nonumber\\
	&= \M{W_\theta}+\Res{\comp}{W_\theta}+\Res{\ti}{V_{\theta\comp}} \label{exchangedecomp}
	\end{align}

	\subsubsection*{Judgement: Covariance between residuals}
	We might try to learn about the global effects covariance matrix $\mathcal{S}_{\theta}$ in full generality. However in practice this proves difficult. Instead we choose to express  $\mathcal{S}_{\theta}$ in terms of known correlation $\Pi_{\theta}$ and unknown second order exchangeable variances $\M{W_\theta}$ and $\Res{\comp}{W_\theta}$. The elements $\mathcal{S}_{\theta_{\comp\comp'}}$ of the global effects covariance matrix $\mathcal{S}_{\theta}$ are given by:
		\begin{align}
	\mathcal{S}_{\theta_{\comp\comp'}}&=W_{\theta\comp}^\frac{1}{2}W_{\theta\comp'}^\frac{1}{2}\Pi_{\theta\comp\comp'}\nonumber\\
	&=(\M{W_\theta}+\Res{\comp}{W_\theta})^\frac{1}{2}(\M{W_\theta}+\Res{\comp'}{W_\theta})^\frac{1}{2}\Pi_{\theta\comp\comp'} \label{eq:sigthet}
	\end{align}
where
	\[\mathrm{corr}(\epithet{\comp \ti},\epithet{\comp' \ti'}) =\left\{\begin{array}{ll}\Pi_{\theta\comp\comp'}  & \textrm{for } \ti=\ti' \\ 0 & \textrm{for } \ti \neq \ti' \end{array} \right.\] 
	
The prior expectation $E(\mathcal{S}_{\theta_{\comp\comp'}})=\Sigma_{\theta_{\comp\comp'}}$. To estimate $\mathcal{S}_{\theta}$ using a sample, we must estimate $\M{W_\theta}$ and $\Res{\comp}{W_\theta}$. In practice, where it is too difficult to learn about $\Res{\comp}{W_\theta}$ we also treat it as known. To update beliefs about $\mathcal{S}_{\theta}$, we therefore only need to learn about $\M{W_\theta}$.
	
	\subsection{Example: Corrosion framework} \label{crrsnfrmwrk}
	
	For the motivating corrosion application (section \ref{data}), ultrasonic and radiographic inspection typically generates values for the minimum wall thickness (or maximum pit depth) corresponding to the area inspected (the ``inspection footprint''). We therefore adopt the specific model framework:
	\begin{align}
	\textrm{Observation Process:}& &\Y_{\comp\ti} &= \min_\loc\left( Z_{\comp\ti} + \epiy{\comp\ti} \right)\\
	\textrm{True System State:} & & Z_{\loc\comp\ti} &= X_{\comp \ti}+ \lc_{\loc \comp \ti} \\
	\textrm{General Corrosion Model:}& & X_{\ti} &= X_{\ti-1} +\CR_{\ti} +\epix{\ti} \nonumber\\
	 & & \CR_{\ti} &= \CR_{(\ti-1)} +\epialp{\ti}\\
	\textrm{Local Corrosion Model:}& & \lc_{\loc \comp \ti} &= \lc_{\loc \comp (\ti-1)} + \epilc{\loc \comp \ti}
	 \end{align}
	where $Z_{\comp\ti}=(Z_{\loc\comp\ti})$ and $\epiy{\comp\ti}$ are vectors over locations. With reference to section \ref{genfrmwrk}, $f$ is the minimum function over set $L$, $g$ is the identity function, and functions $F$, $G$  and parameter vector $\theta_{\ti}$ take the form:
	\begin{align*}
	  F&= \left( \begin{array}{cc}
	                 I_C & 0_C \\
	               \end{array}
	             \right)&
	\thet_\ti &=\left(\begin{array}{c}
	             \X_{\ti} \\
	             \CR_{\ti}
	           \end{array} \right) &
	  G&= \left(\begin{array}{cc}
	                 I_C & I_C \\
	                 0_C & I_C
	               \end{array}
	             \right)
	\end{align*}
where $I_C$ and $0_C$ are $C \times C$ identity and zero matrices respectively, and $X_t$ is a vector of wall thicknesses over components. We seek inferences about true system state $Z_{\loc\comp\ti}$ and corrosion rate vector $\CR_{\ti}$ over components. The error terms $\epix{\comp \ti}$, $\epialp{\comp \ti}$ and $\epilc{\loc\comp \ti}$ control evolution of system level, slope and local level adjustment respectively.
	
We make the analogous exchangeability assumptions as in section \ref{exchangassmp} so that:
	\begin{align*}
	\epix{\comp \ti}^2 &= \M{W_\X}+\Res{\comp}{W_\X}+\Res{\ti}{V_{\X\comp}} \\
	\epilc{\comp \ti}^2 &= \M{W_\CR}+\Res{\comp}{W_\CR}+\Res{\ti}{V_{\CR\comp}}
	\end{align*}
as given in equation \ref{exchangedecomp}. In principle it is possible to learn about $\epialp{\comp\ti}^2$ and $\epix{\comp\ti}^2$ separately.  In practice, given sparsity of data, this can be difficult. For simplicity we fix the ratio of mean variances:
\begin{equation}\M{W_{\CR}}) = \lambda \M{W_{\X}} \label{varratio}\end{equation}
The correlation structure of wall thickness evolution error $\Pi_{\X\comp\comp'}$ is assumed to take the form of a linear combination of three terms, reflecting an underlying universal correlation $\rho_0$ between all pairs of components (regardless of the circuit(s) to which they correspond), a circuit correlation $\rho_C$ between all pairs of components within the same circuit,  and a correlation $\rho_D$ which decays exponentially at rate $\nu>0$ with distance $s$ (measured in terms of the number of intervening components along the circuit between the components). So the covariance between components $\comp$ and $\comp'$ is:
\begin{equation}
\Pi_{\X\comp\comp'}=\rho_0+\rho_C \delta_{\comp \comp'} +\rho_D  e^{-\nu s_{\comp\comp'}} \label{eq:corrpi}
\end{equation}
where $\delta_{\comp \comp'}=1$ if components $\comp$ and $\comp'$ are in the same circuit, and $s_{\comp\comp'}$ is the distance between the components. For simplicity we assume $\Pi_\CR=\Pi_\X$ and then:
\[\Pi_\theta=\left(\begin{array}{cc} \Pi_\X & 0 \\ 0 & \Pi_\CR\end{array}\right)\]
such that $\mathcal{S}_\X$ and $\mathcal{S}_\CR$ are defined using equation \ref{eq:sigthet}. The correlation matrix $\Pi_\X$ for the offshore application is illustrated in figure \ref{corrmat}.
	
\begin{figure}
\includegraphics[width=\textwidth]{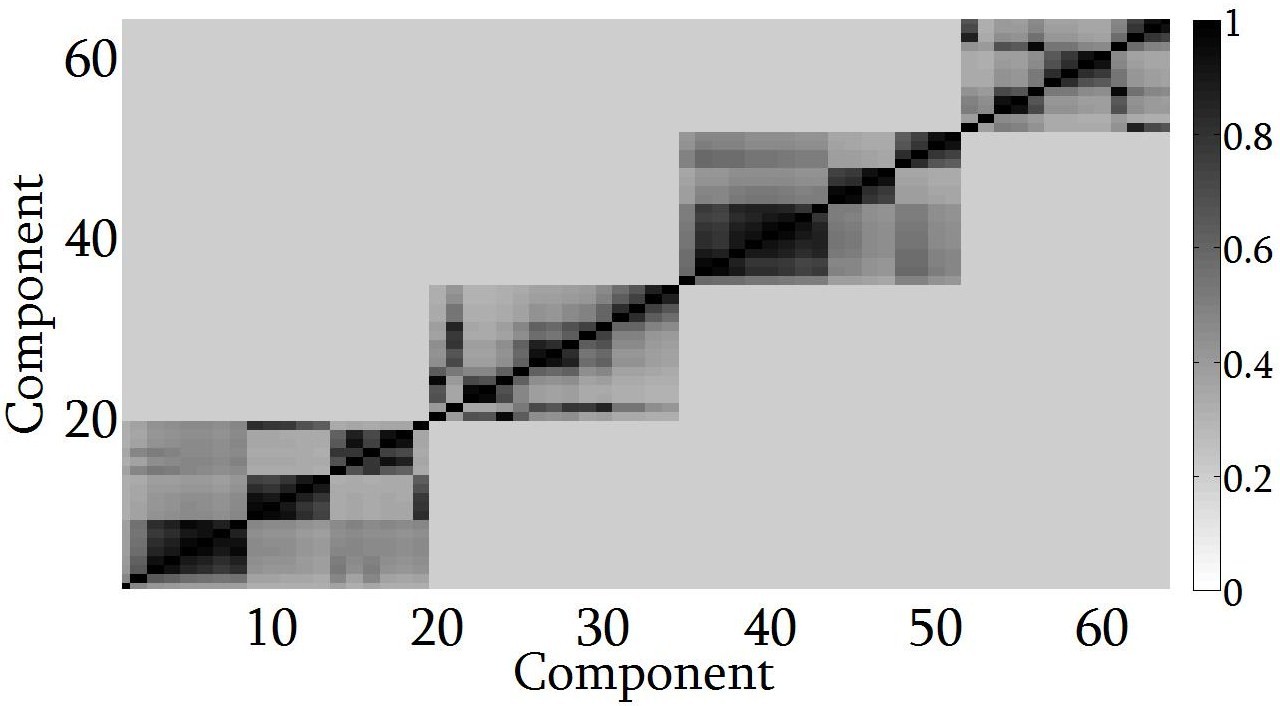}
\caption{Correlation matrix $\Pi_\X$ for the offshore application. The 4 blocks present correspond to the four corrosion circuits.} \label{corrmat}
\end{figure}
	
\section{Inference}  \label{BLinf}
	
The objective of the analysis is to learn about the characteristics of model parameters in equations 9-12. We achieve this using an iterative scheme explained in sections 5.1-5.4. We use Bayes linear mean updating in section \ref{BLmean} to learn about the true system state $Z_{\comp\ti}$ and global effects parameter vector $\theta_t$ assuming that all variance parameters in the model are known.  In section \ref{BLvar} we show how Bayes linear variance learning is use to adjust beliefs about global effects model variance parameter $\M{W_\theta}$ assuming known local effects variance parameters. In section \ref{mahal} we use Mahalanobis variance learning to estimate local effects variance parameters $\Sigma_\lc$ given updated global effects model variance $\M{W_\theta}$.  Model fit and prior assumptions are assessed using discrepancy diagnostics discussed in section 5.4.
	
\subsection{Bayes linear mean updating} \label{BLmean}
	
Given a collection of observations of one or more system components at one or more times expressed in vector form as $\Y$, we adjust beliefs about true system state $Z^{min}_{\comp\ti}$ and global effects level $\theta_{\ti}$ by calculating $E_\Y(Z^{min}_{\comp\ti})$ and $E_\Y(\theta_{\comp\ti})$:
	\begin{align}
	E_\Y(Z^{min}_{\comp\ti}) &= E(Z^{min}_{\comp \ti})+\cov(Z^{min}_{\comp \ti},\Y)[\var(\Y)]^{\dagger}(\Y-E(\Y))\label{eq:BLmeanupZ} \\
	 E_\Y(\theta_{\comp\ti}) &= E(\theta_{\comp\ti})+\cov(\theta_{\comp\ti},\Y)[\var(\Y)]^{\dagger}(\Y-E(\Y)) \label{eq:BLmeanupthet}
	\end{align}
	where $Z^{min}_{\comp\ti}=\min_l(Z_{\comp\ti})$ over locations. In a similar fashion we can calculate adjusted variances:
	\begin{eqnarray*}
	\var_\Y(Z^{min}_{\comp \ti}) &=& \var(Z^{min}_{\comp \ti})-\cov(Z^{min}_{\comp \ti},\Y)[\var(\Y)]^{\dagger}\cov(\Y,Z^{min}_{\comp \ti})\\
	 \var_\Y(\theta_{\comp\ti}) &=& \var(\theta_{\comp\ti})-\cov(\theta_{\comp\ti},\Y)[\var(\Y)]^{\dagger}\cov(\Y,\theta_{\comp\ti})
	\end{eqnarray*}
	
	Estimates for (co-)variances $\var(\Y)$, $\var(Z_{\loc\comp \ti})$, $\var(\theta_{\comp \ti})$, $\cov(Z_{\loc\comp \ti},\Y)$ and  $\cov(\theta_{\comp \ti},\Y)$ for corresponding components and times are obtained by simulation under the model using prior beliefs. Details of the simulation procedure are given in section \ref{alg}. Briefly we simulate realisations for the complete system from the model given in section \ref{model} given prior specification. Realisations are then use to calculate empirical estimates for expectations, variances and covariances used in Bayes linear adjustment. We can use simulation to run the model forward in time beyond the period of the data or to predict characteristics of unseen components in a straight forward manner.
	
\subsection{Bayes linear variance learning} \label{BLvar}
	
As illustrated in section 4.2, the Bayes linear approach may also be used for variance learning. Here, development of Bayes linear variance learning is considered for the corrosion modelling application for concreteness, since the corrosion application illustrates all key modelling features. Other applications will require modification of calculation details in general, but the methodology remains applicable. We seek expressions for squared residuals $\epix{\comp \ti}^2$ from sample data corresponding to partial system inspections with which to learn about population mean variances $\M{W_{\X}}$, and thereby $\M{W_{\alpha}}$ and $\Sigma_r$ also. Since the general corrosion DLM part of the model is invertible we can take linear combinations of observations to isolate expressions for $\epix{\comp \ti}^2$ even when observations of the system are irregularly spaced in time. Using expressions for $\epix{\comp \ti}^2$ thus obtained, we adjust our beliefs about $\M{W_{\X}}$ using observed data and assumed known values for $\Sigma_r$ as explained below. In section \ref{mahal} we present a fitting procedure using a Mahalanobis distance criterion to select an optimal combination of error variances.
	
\subsubsection*{Motivation}
	
We motivate the theorem below  by considering the case of full regularly spaced inspections, with the observation equation expressed as:
	\begin{eqnarray}
	\Y_{\comp \ti} &=& \X_{\comp \ti}+M_{\comp \ti} \label{obs_wall_thick}
	\end{eqnarray}
where:
\[M_{\comp \ti}=\min_{\loc} \left(\lc_{\loc\comp \ti} + \epiy{\loc\comp \ti}  \right)\]
consider taking differences of observations in time. Let $\Y_{\comp \ti}^{(i)}$ be the $i$th step difference for component $\comp$ at time $\ti$ defined by:
	\begin{eqnarray*}
	\Y_{\comp \ti}^{(i)}=\Y_{\comp \ti}-\Y_{\comp(\ti-i)}
	\end{eqnarray*}
with the analogous definition for $M_{\comp \ti}^{(i)}$. Then using equations 9 and \ref{obs_wall_thick}, one and two step differences are seen to be:
\begin{eqnarray*}
\Y_{\comp \ti}^{(1)}=\Y_{\comp \ti}-\Y_{\comp(\ti-1)}&=&\X_{\comp \ti}-\X_{\comp(\ti-1)}+M_{\comp \ti}-M_{\comp(\ti-1)}\\
&=&\X_{\comp(\ti-1)}+\CR_{\comp \ti}+\epix{\comp \ti}-\X_{\comp(\ti-1)}+M_{\comp \ti}^{(1)}\\
&=& \CR_{\comp \ti}+\epix{\comp \ti}+M_{\comp \ti}^{(1)}\\
&=& \CR_{\comp(\ti-1)}+\epialp{\comp \ti}+\epix{\comp \ti}+M_{\comp \ti}^{(1)}
\end{eqnarray*}
and:
\begin{align*}
\Y_{\comp \ti}^{(2)}=\Y_{\comp \ti}-\Y_{\comp(\ti-2)}&=\X_{\comp \ti}-\X_{\comp(\ti-2)}+M_{\comp \ti}^{(2)}\\
&= 2\CR_{\comp(\ti-1)}+\epialp{\comp \ti}+\epix{\comp(\ti-1)}+\epix{\comp \ti}+M_{\comp \ti}^{(2)}
\end{align*}
The linear combination $\Y_{\comp \ti}^{(2)}-2\Y_{\comp \ti}^{(1)}$ isolates the residual error structure except for the term $M_{\comp \ti}^{(2)}-2M_{\comp \ti}^{(1)}$:
	\[\Y_{\comp \ti}^{(2)}-2\Y_{\comp \ti}^{(1)}=-\epialp{\comp \ti}+\epix{\comp(\ti-1)}-\epix{\comp \ti}+M_{\comp \ti}^{(2)}-2M_{\comp \ti}^{(1)}\]
The expectation of the square of the former linear combination is found using the facts that residuals are mutually uncorrelated in time, together with the exchangeability assumptions from equation\ref{exchangedecomp} and the assumption of a fixed ratio of population variances in equation \ref{varratio}. Then:
	\begin{align}
	E(\Y_{\comp \ti}^{(2)}-2\Y_{\comp \ti}^{(1)})^2 =& E\left(\epialp{\comp \ti}+\epix{\comp(\ti-1)}-\epix{\comp \ti}+M_{\comp \ti}^{(2)}-2M_{\comp \ti}^{(1)}\right)^2\nonumber\\
	=& E\left(\epialp{\comp \ti}^2+\epix{\comp(\ti-1)}^2+\epix{\comp(\ti)}^2+(M_{\comp \ti}^{(2)}-2M_{\comp \ti}^{(1)})^2\right)\nonumber\\
	=& E\bigg( \lambda \big(\M{W_\X} + \Res{\comp}{W_\X}+\Res{\ti}{V_{\X\comp}}\big) \bigg. \nonumber\\
	&+\M{W_\X}+\Res{\comp}{W_\X}+\Res{\ti-1}{V_{\X\comp}} \nonumber\\
	& + \bigg.\M{W_\X}+\Res{\comp}{W_\X}+\Res{\ti}{V_{\X\comp}}+(M_{\comp \ti}^{(2)}-2M_{\comp \ti}^{(1)})^2 \bigg)\nonumber\\
	=& E \left((\lambda+2)\M{W_\X}\right)+E\left((M_{\comp \ti}^{(2)}-2M_{\comp \ti}^{(1)})^2 \right)\nonumber\\
	=& (\lambda+2)\musigX+E\left((M_{\comp \ti}^{(2)}-2M_{\comp \ti}^{(1)})^2\right) \label{fullinspdiff}
\end{align}

Using the corresponding result for irregular inspection, outlined in appendix \ref{Edbar}, we arrive at the following theorem.

	\begin{thm} \label{thm1}
	Consider incomplete inspection of a system yielding observations $\{ \Y_{\comp \ti^{\comp}_1}, \Y_{\comp \ti^{\comp}_2}, \dots, \Y_{\comp \ti^{\comp}_{T_\comp}} \}$ at times $\{ \ti^{\comp}_1, \ti^{\comp}_2, \dots, \ti^{\comp}_{T_\comp}\}$ for components $\comp=1,2,\dots,C$.
	\noindent The adjusted expectation $E_{\Dbar}(\M{W_\X})$ for $\M{W_{\X}}$ given $\Dbar$ is:
	
	\[E_{\Dbar}(\M{W_\X})=E(\M{W_\X}) +\cov(\M{W_\X},\Dbar)(\var(\Dbar))^{\dagger}(\Dbar-E(\Dbar))\]
\noindent	where $\Dbar$ is a $\by{C}{1}$ vector over components with elements $\Dbar_{\comp}$:
	\[ \Dbar_{\comp} = \sum_{i=3}^{T_\comp} (k_i^c Y_{\comp \ti^{\comp}_i}^{(2)}-l_i^c Y_{\comp \ti^{\comp}_i}^{(1)})^2 \]
Further, $E(\M{W_\X})=\musigX$,  $\cov(\M{W_\X},\Dbar)=1_C\trans (T_\comp-2)\covsigX$ and:
	\begin{align*}
			E(\Dbar_{\comp i})=&\frac{k_i^c l_i^c (k_i^c -l_i^c) (2 \lambda (k_i^c)^2 -2\lambda l_i^c -6 -\lambda)}{6} \musigX\\
			& +(l_i^c)^2 E((M_{\comp \ti^{\comp}_i}^{(1)})^2)+ (k_i^c)^2 E((M_{\comp \ti^{\comp}_i}^{(2)})^2)- 2k_i^c l_i^c E(M_{\comp \ti^{\comp}_i}^{(2)}M_{\comp \ti^{\comp}_i}^{(1)})
	\end{align*}
	\noindent where $k_i^c = \ti^{\comp}_i - \ti^{\comp}_{i-1}$, $l_i^c = \ti^{\comp}_i - \ti^{\comp}_{i-2}$ and $Y_{\comp \ti^{\comp}_i}^{(j)} = Y_{\comp \ti^{\comp}_i}-Y_{\comp \ti^{\comp}_{(i-j)}}$ with a similar definition for $M_{\comp \ti^{\comp}_i}^{(j)}$.
	\end{thm}
	
	\begin{proof}
	$E(\M{W_\X})=\musigX$ by definition in section \ref{wallthickevo}. Derivations of expressions for $E(\Dbar)$ and $\cov(\M{W_\X},\Dbar)$ are given in appendix \ref{Edbar} and \ref{cov}
	\end{proof}
	
	\begin{cor}
	In the case of full regular inspection yielding observations $\Y_{\comp \ti}$ at times $\ti=1,2, \dots,T$ for components $\comp$, the expression for $\Dbar_{\comp}$ is:
	\[\Dbar_{\comp}=\sum_{t=3}^T \left(\Y_{\comp \ti}^{(2)}-2\Y_{\comp \ti}^{(1)}\right)^2\]
	\end{cor}
	\begin{proof}
	Let $k^\comp_i=1$ and $l^\comp_i=2$ for each $\ti^\comp_i$ in theorem \ref{thm1} and $E(\Dbar)$ is found as in equation \ref{fullinspdiff}
	\end{proof}	
	
	\subsubsection{Corrosion Example}
For known local corrosion variance, $\Sigma_r$, we use simulation with prior beliefs to estimate $E[(M_{\comp \ti^{\comp}_i}^{(1)})^2]$, $E[(M_{\comp \ti^{\comp}_i}^{(2)})^2]$, and $E[M_{\comp \ti^{\comp}_i}^{(2)}M_{\comp \ti^{\comp}_i}^{(1)}]$ and $\var(\Dbar))^{\dagger}$. Details of simulation procedure are given in section \ref{alg}. As illustration, consider estimating $\M{W_\X}$ using $E_{\Dbar}(\M{W_\X})$ for a simulated case using the actual inspection design from section 2, with $\M{W_\X}$ set to $0.1^2$. For each of $50$ independent realisations of the systems over all components and times, we calculate $\Dbar$ at the inspection times, which is then used to calculate $E_{\Dbar}(\M{W_\X})$. The mean value of $E_{\Dbar}(\M{W_\X})$ is $0.0994^2$, with empirical $5\%$ and $95\%$ values (from simulation) of respectively $0.0872^2$ and $0.1128^2$, consistent with the known value.
	
This procedure allows variance learning for the general corrosion DLM when local corrosion variances are known. In practice this is not the case. Selection of optimal combination of local and general variance estimates consistent with data is discussed next.
	
\subsection{Mahalanobis variance learning} \label{mahal}
	
As for general corrosion in section \ref{BLvar}, we would ideally use a Bayes linear scheme to update local corrosion variances also. However the non-linear nature of the observation equation (equation 9) renders direct estimation impossible.  Instead we adopt a fitting procedure based on Mahalanobis distance (\citet{mah36}), exploiting the estimated covariance structure, to estimate combinations of local and general corrosion error variances consistent with observational data, given an estimate of measurement error variance.
	
For each of a set of $p$ candidate values for local corrosion variance $\Sigma_\lc$, we calculate $E_{\Dbar}(\M{W_\X})$ by Bayes linear adjustment. Adopting $E_{\Dbar}(\M{W_\X})$ as an updated estimate for $\M{W_\X}$, and $\lambda E_{\Dbar}(\M{W_\X})$ for $\M{W_\CR}$, we re-simulate to estimate $E(\Y)$, $\var(\Y)$ and the ratio of Mahalanobis distance to its expected values, termed the discrepancy ratio, $H$:
\begin{equation}
H = \frac{(\Y- E(\Y))\trans \var(\Y)^{\dagger} (\Y - E(\Y))}{\mathrm{rank}(\var(\Y))} \label{discratio}
\end{equation}
	as outlined  in section \ref{alg}.	We select the candidate $\Sigma_\lc$ which yields a discrepancy ratio nearest to its expected value of unity.
	\begin{figure}
	\includegraphics[width=\textwidth]{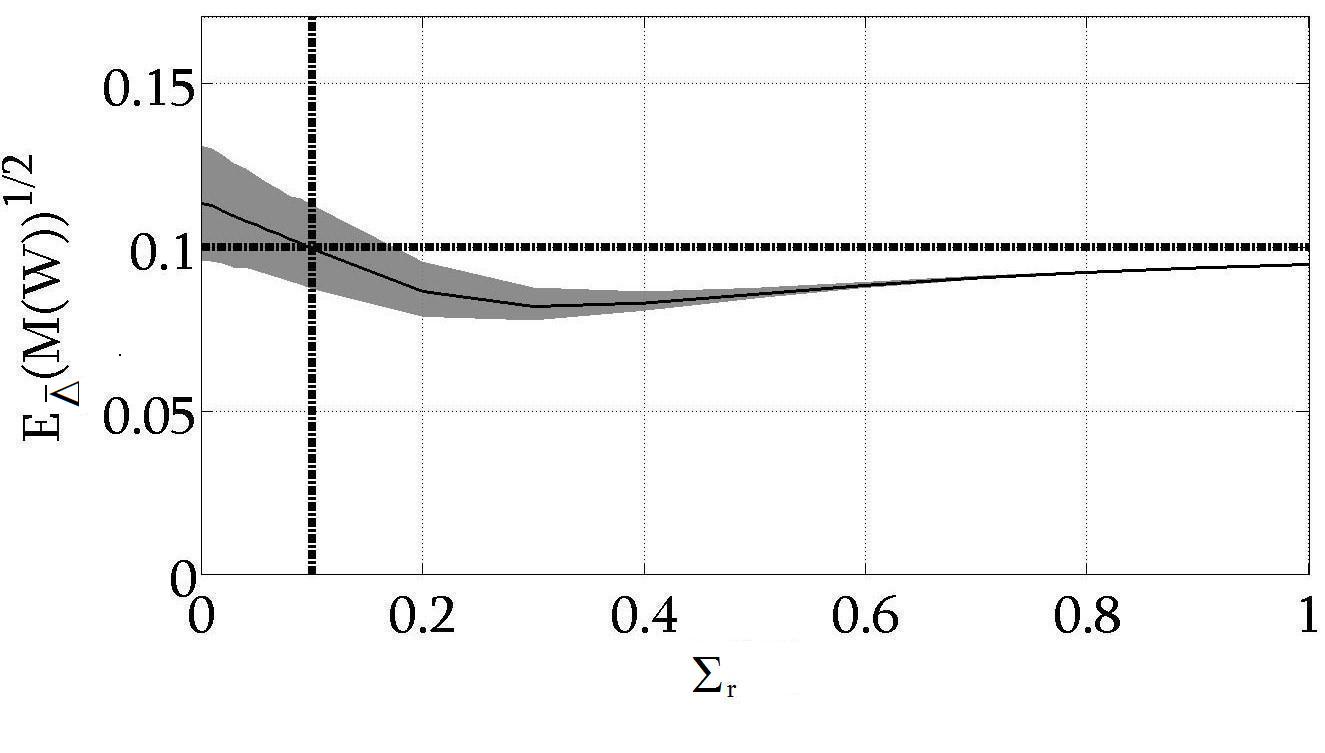}
	\caption{Bayes linear variance learning for wall thickness variance $E_{\Dbar}(\M{W_\X})$ as a function of local corrosion variance, $\musiglc$, for simulated data, shown on square-root scale. The true values of $\M{W_\X}$ (for which $E_{\Dbar}(\M{W_\X})$ is an estimate) and $\musiglc$ are both $0.1^2$, as shown by the dashed horizontal and vertical lines. The mean estimate for $E_{\Dbar}(\M{W_\X})^{1/2}$ is shown as a solid line, and the shaded region corresponds to a 90\% uncertainty band for $E_{\Dbar}(\M{W_\X})^{1/2}$ bounded by the $5$th and $95$th percentiles derived from simulation.} \label{localvsgeneralbounds}
	\end{figure}
	\begin{figure}
	\includegraphics[width=\textwidth]{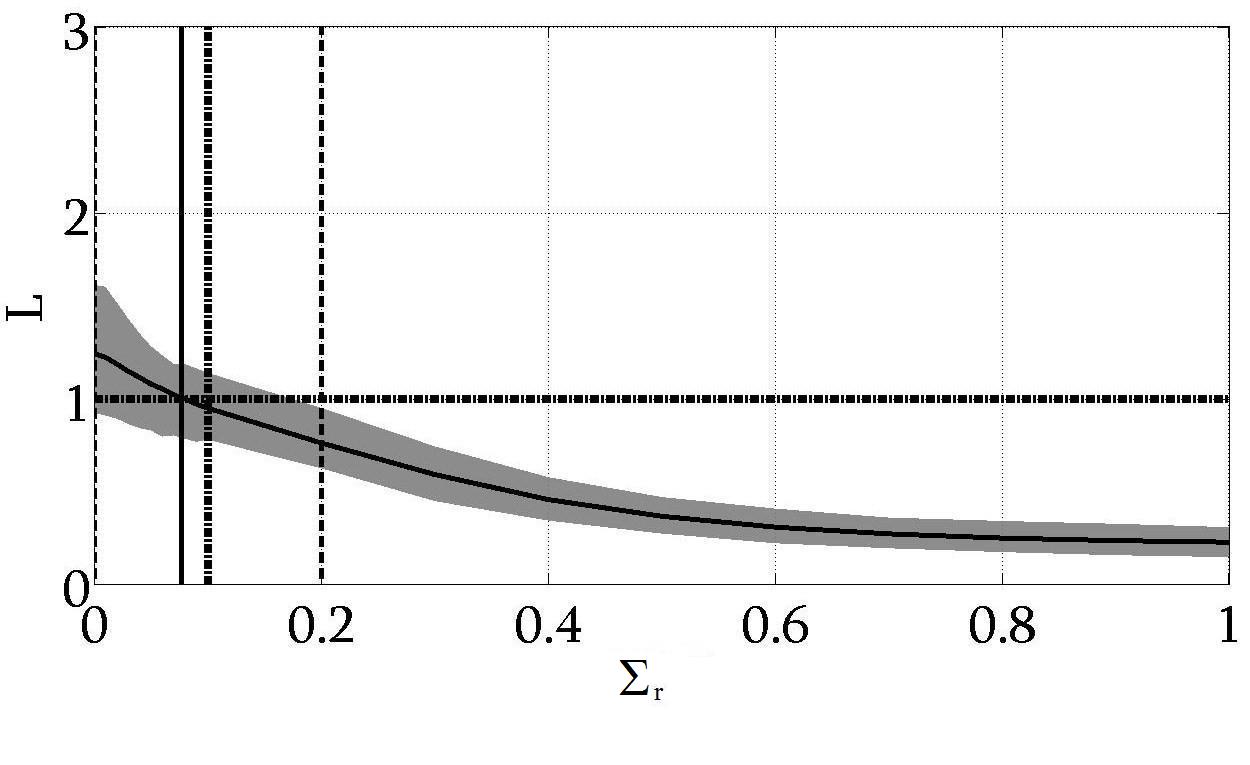}
	\caption{Discrepancy ratio, $H$, as a function of local corrosion variance $\musiglc$ for simulated data. The mean estimate for $H$ is shown as a solid line, and the shaded region corresponds to a 90\% uncertainty band for $H$ bounded by the 5th and 95th percentiles derived from simulation. The expected value of $H$ shown as a horizontal line is 1, suggesting $\musiglc=0.08^2$ and corresponding $E_{\Dbar}(\M{W_\X})=0.11^2$, (from figure \ref{localvsgeneralbounds}). Vertical lines indicate the mean (solid) and 5th and 95th percentiles (dashed) for the particular choice of $\musiglc$ in individual realisations. The true values of $\M{W_\X}$ and $\musiglc$ are both $0.1^2$.}\label{Mahalanobisbounds}
	\end{figure}
	We illustrate joint Bayes linear and Mahalanobis variance learning for the simulated example from section 5.2.1 in figures \ref{localvsgeneralbounds} and \ref{Mahalanobisbounds}. Figure \ref{localvsgeneralbounds} shows Bayes linear variance learning for wall thickness variance $E_{\Dbar}(\M{W_\X})$ as a function of local corrosion variance $\musiglc$. The simulated example has true values of $\M{W_\X}$ and $\musiglc$ both $0.1^2$. Figure \ref{Mahalanobisbounds} shows the discrepancy ratio $H$ as a function of local corrosion variance $\musiglc$ for the same example. Comparison with the expected value of $H$, shown as a horizontal line, suggests the estimate $\musiglc=0.08^2$. Figure \ref{localvsgeneralbounds} then provides the corresponding estimate $E_{\Dbar}(\M{W_\X})=0.11^2$.
	
	\subsection{Diagnostics} \label{diagdis}
	Diagnostics are essential to confirm the adequacy of model fit. Using Mahalanobis distance, an adjustment discrepancy $\mathrm{Dis}_Y(X)$ (see page 122 of \citet{BLS07}) for adjusted expectations can be calculated:
	\[\mathrm{Dis}_Y(X)=(E_{Y}(X)-E(X))\trans\mathrm{RVar}_Y{X}^{-1}(E_{Y}(X)-E(X))\]
where $X$ is a vector of wall thicknesses for all components and times and $\mathrm{RVar}_Y{X}$ is defined in equation \ref{Rvar}. Analogously, $\mathrm{Dis}_Y(Z)$ and $\mathrm{Dis}_Y(\CR)$ can be calculated for vectors $Z$ and $\CR$ over all components and times. We can also compute discrepancy directly on the data (see page 100 of \cite{BLS07}):
	\begin{equation}\mathrm{Dis} (Y)=(Y-E(Y))\trans\mathrm{Var}(Y)^{-1}(Y-E(Y)) \label{eq:Disy}\end{equation}
	The expected value of both $\mathrm{Dis}_Y(X)$ and $\mathrm{Dis} (Y)$ is unity. Forms for both discrepancies can also be evaluated for subsets of variables, useful, e.g., to spot individual outliers. Discrepancies are calculated on the prior specification, adjusted expectations of global effects terms and, subsequent to learning about local effect variance parameters, quantifying model adequacy at each stage. General thresholds for discrepancy measures do not exist although as a heuristic the three sigma rule (\citealt{Puk94}) states that for any uni--modal continuous random quantity, $X$, $P(|X-E(X)| \leq 3\sqrt{\var(X)}) \geq 0.95$.
	\begin{figure}
	\includegraphics[width=\textwidth]{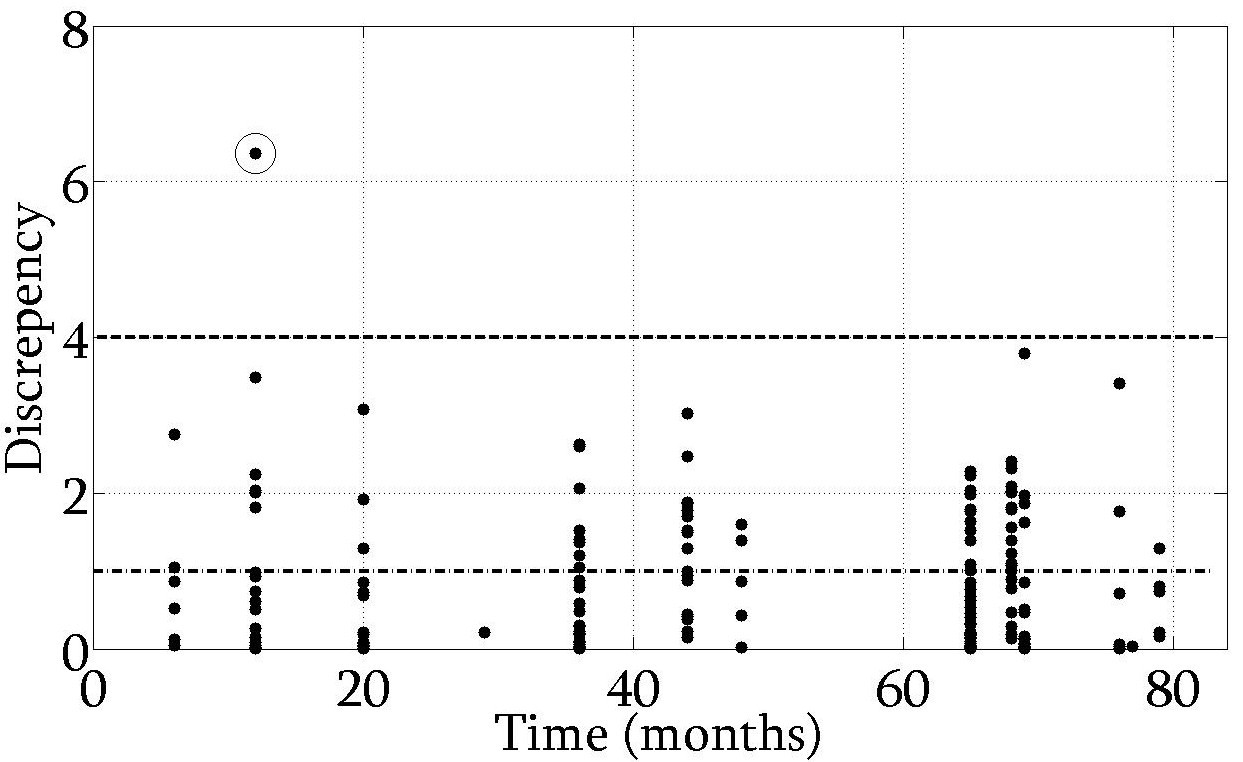}
	\caption{Component-wise discrepancy $\mathrm{Dis}(Y)$  for a typical realisation from the simulated data. The expected value of $\mathrm{Dis}(Y)$ is unity, shown as a horizontal line. Also shown is the horizontal line corresponding to $|1-\mathrm{Dis}(Y)|=3$, serving as a warning limit for unusually large values of discrepancy. In this realisation, only one observation exceeds the warning limit. } \label{DataDiscrp2}
	\end{figure}
	Figure \ref{DataDiscrp2} shows $\mathrm{Dis}(Y)$ for a realisation of the simulated data discussed in sections 5.2 and 5.3. The expected value of $\mathrm{Dis}(Y)$ is unity, shown as a horizontal line. In this realisation, only one observation exceeds the warning limit suggesting that prior specification and data are consistent. More general diagnostics for variance learning and overall model fit can be specified similarly.
	
\section{Inference procedure} \label{alg}
	
The procedure to apply the corrosion model (section \ref{crrsnfrmwrk}) is explained below. Firstly we make prior specification and carry out simple diagnostic checks to confirm consistency of the data and priors. We then update model (co-)variances, before updating means. Finally we calculate diagnostics to assess model fit.
	
To simulate from the model, distributional assumptions are necessary below (indicated using notation $\mathrm{iid}(.)$. It is advised that a variety of different distributional forms be examined to enable an informed choice about (co-)variance specification. Plausible distribution will generally lead to robust estimates of low order moments, even if full probability specification would not be appropriate using these distributions.
	
	\begin{enumerate}

	\item Specify prior and starting values to be examined
	\begin{itemize}
	\item wall thickness variances prior, $\musigX$, $\varsigX$ and $\covsigX$ (see section \ref{exchangassmp}), and $\Pi_{X_{\comp\comp'}}$ (see equation \ref{eq:corrpi})
	\item corrosion variance variance prior, $\lambda$, $\Pi_{\CR_{\comp\comp'}}$ (see section \ref{crrsnfrmwrk})
	\item measurement error variance $\Sigy$ (see section \ref{genfrmwrk})
	\item candidate values for $p$ local corrosion variances $\Sigma_\lc$ to consider (see section \ref{genfrmwrk})
	\item starting values over components for $\X_t$ and $\alpha_t$ at $t=0$ (see section \ref{crrsnfrmwrk})
	\end{itemize}
	
	\item For each of $p$ choices of local corrosion variance $\Sigma_\lc$ (see section \ref{genfrmwrk})
	\begin{enumerate}
		\item \label{simvar} \label{step2} Evaluate variance matrices using simulation under prior specification 
    	\begin{itemize}
	\item simulate $\mathcal{S}_\X$ from $\mathrm{iid}(\musigX,\varsigX, \covsigX,\Pi_{X_{\comp\comp'}})$
	\item simulate $\mathcal{S}_\CR$ from $\mathrm{iid}(\lambda\musigX,\lambda^2\varsigX, \lambda^2\covsigX,\Pi_{\CR_{\comp\comp'}})$
	\end{itemize}	
	\item Run model specified in section \ref{crrsnfrmwrk} forward from $t=0$ \label{simvar2}\label{step3}
	\begin{itemize}
	\item $\epialp{.\ti} \sim \mathrm{iid}(0,  \mathcal{S}_\CR)$
	\item $\CR_{.\ti} = \CR_{.\ti-1}+\epialp{.\ti}$
	\item $\epix{.\ti} \sim \mathrm{iid}(0,  \mathcal{S}_\X)$
	\item $\X_{.\ti} = \X_{.\ti-1}+\CR_{.\ti}+\epix{.\ti}$
	\item $\epilc{..\ti} \sim \mathrm{iid}(0,  \Sigma_\lc)$
	\item $\lc_{. .\ti} = \lc_{. .\ti-1}+\epilc{. .\ti}$
	\item $\epiy{..\ti} \sim \mathrm{iid}(0,  \Sigma_\Y)$
	\item $\Y_{. .\ti} =\min_{\loc}\left( \X_{.\ti}+\lc{. .\ti-1}+\epiy{. .\ti} \right)$
	\end{itemize}
	
	\item From simulation calculate sample estimates $E(\Y)$ and $\var(\Y)$ over $N$ realisations and check consistency with prior by computing discrepancy $\mathrm{Dis}(\Y)$ (see equation \ref{eq:Disy})
		\item Calculate $\Dbar$ and $\var(\Dbar)$  over $N$ realisations and calculate $E_{\Dbar}(\M{W_\X})$ for each $\Sigma_\lc$ (see theorem \ref{thm1})
		
		\end{enumerate}
		
	\item Select optimal $\Sigma_\lc$ from the set of $p$ candidates as that which yields the discrepancy ratio (see equation \ref{discratio}) nearest to unity
	\item Simulate to update wall thicknesses $\X_\ti$ and corrosion rates $\CR_\ti$ for any $\ti$ of interest
	\begin{itemize}
		\item repeat step \ref{simvar} and \ref{simvar2} using optimal $\Sigma_\lc$ and $\musigX=E_{\Dbar}(\M{W_\X})$ for $N$ realisations
		\item calculate sample estimates from simulation for $E(Z^{min}_{\comp\ti})$, $E(\CR_{\comp\ti})$, $\cov(Z^{min}_{\comp\ti},\Y)$,$\cov(\CR_{\comp\ti},\Y)$, $\var(\Y)$ and $E(\Y)$
		\item calculate adjusted expectations from true wall thickness $E_\Y (Z^{min}_{\comp\ti})$ and corrosion $E_\Y (\CR_{\comp\ti})$ (see section \ref{BLmean}), for components and times of interest. This could include forecasting forward in time or prediction of components not observed
	\end{itemize}		
\item Calculate discrepancy ratio $H$ (see equation \ref{discratio}) for estimated model variables to confirm consistency.
	\end{enumerate}

	\section{Application} \label{examcorm} \label{diag}
	
	Having illustrated model performance for simulated data, we now apply it to actual historical inspection data for the offshore platform, using the procedure described in section \ref{alg}. The first step of the analysis is specification of prior values. Recalling that the inspection design (figure \ref{siminspdes}) and system prior specification (section \ref{exampleintro}) corresponding to the historical data were also used for simulated data, results equivalent to those in figures \ref{localvsgeneralbounds}, \ref{Mahalanobisbounds} and \ref{DataDiscrp2} are shown in figures \ref{realDataDiscrp2}, \ref{reallocalvsgeneralbounds} and \ref{realMahalanobisbounds} below.
	
	We assess the prior discrepancy of observations as described in section \ref{diagdis}. To achieve this, realisations under the model are generated as described in steps \ref{step2} and \ref{step3} of section \ref{alg}. Discrepancies are shown in figure \ref{realDataDiscrp2}, to be compared with figure \ref{DataDiscrp2} for a single system realisation. With one exception, historical inspection data are consistent with prior expectations. At time point $11$, corresponding to the first observations of components $35-51$, all in a particular corrosion circuit, high values of discrepancy may reflect poor prior specification for that circuit at that time.
	\begin{figure}
	\includegraphics[width=\textwidth]{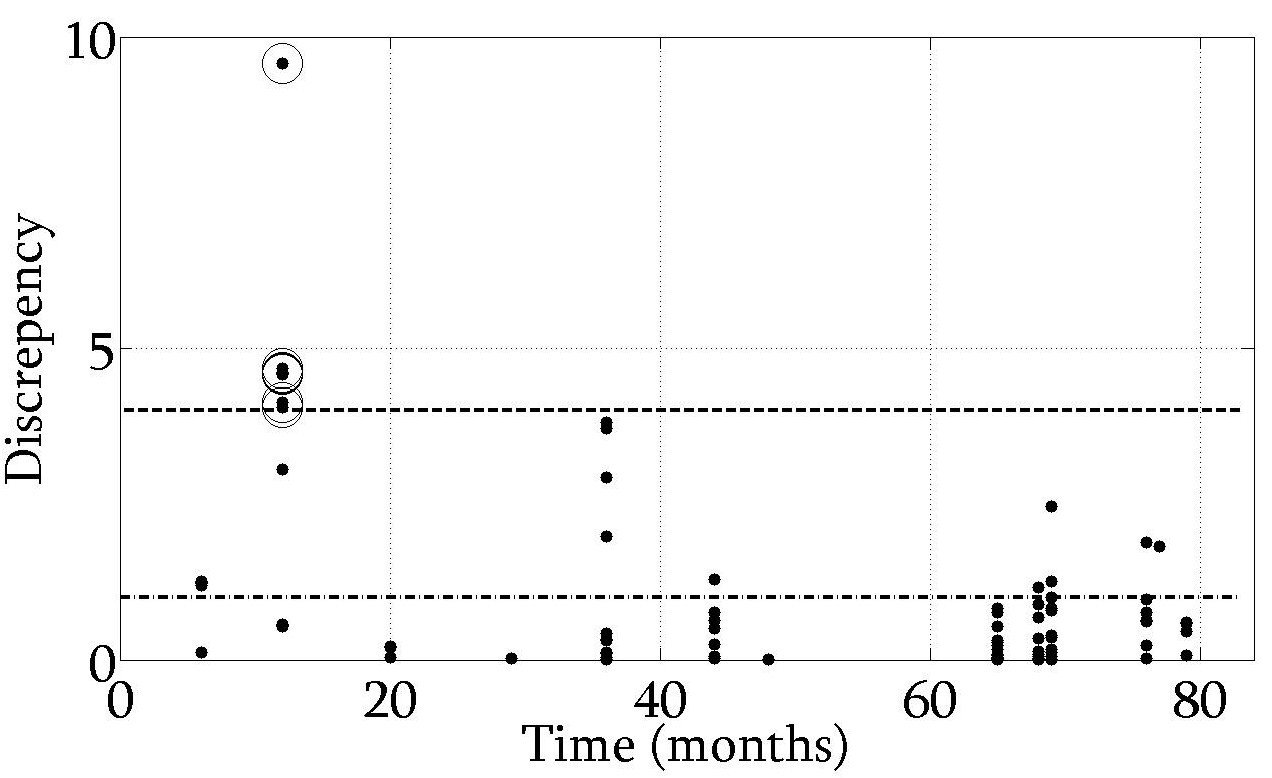}
	\caption{Component-wise discrepancy $\mathrm{Dis}(Y)$  for the historical inspection data. The expected value of $\mathrm{Dis}(Y)$ is unity, shown as a horizontal line. Also shown is the horizontal line corresponding to $|1-\mathrm{Dis}(Y)|=3$, serving as a warning limit for unusually large values of discrepancy. At time point $11$, corresponding to the first observations of components $35-51$, all in a particular corrosion circuit, high values of discrepancy may reflect poor prior specification for that circuit at that time point.} \label{realDataDiscrp2}
	\end{figure}
	Starting values for wall thickness $X_0$ and corrosion $\CR_0$ were specified using historical data. We evaluate adjusted expectation for wall thickness variance $E_{\Dbar}(\M{W_\X})$ as described in theorem \ref{thm1} for different local corrosion variance $\musiglc$ using variance learning (section \ref{BLvar}). Results are given in figures  \ref{reallocalvsgeneralbounds} and \ref{realMahalanobisbounds} (analogous to figures \ref{localvsgeneralbounds} and \ref{Mahalanobisbounds} for the simulated data). The value of $0.1^2$ for $E_{\Dbar}(\M{W_\X})$, to be used as an estimate for prior wall thickness variance $\M{W_\X}$ is the same as its prior value (\ref{priorreal}). However, the updated value of $0.4^2$ for $\musiglc$ is considerably larger than its prior ($0.1^2$). Wall thickness forecasts and inspection designs derived using updated estimates for corrosion variances will be different in general to those based on constant (prior) variances.
	\begin{figure}
	\includegraphics[width=\textwidth]{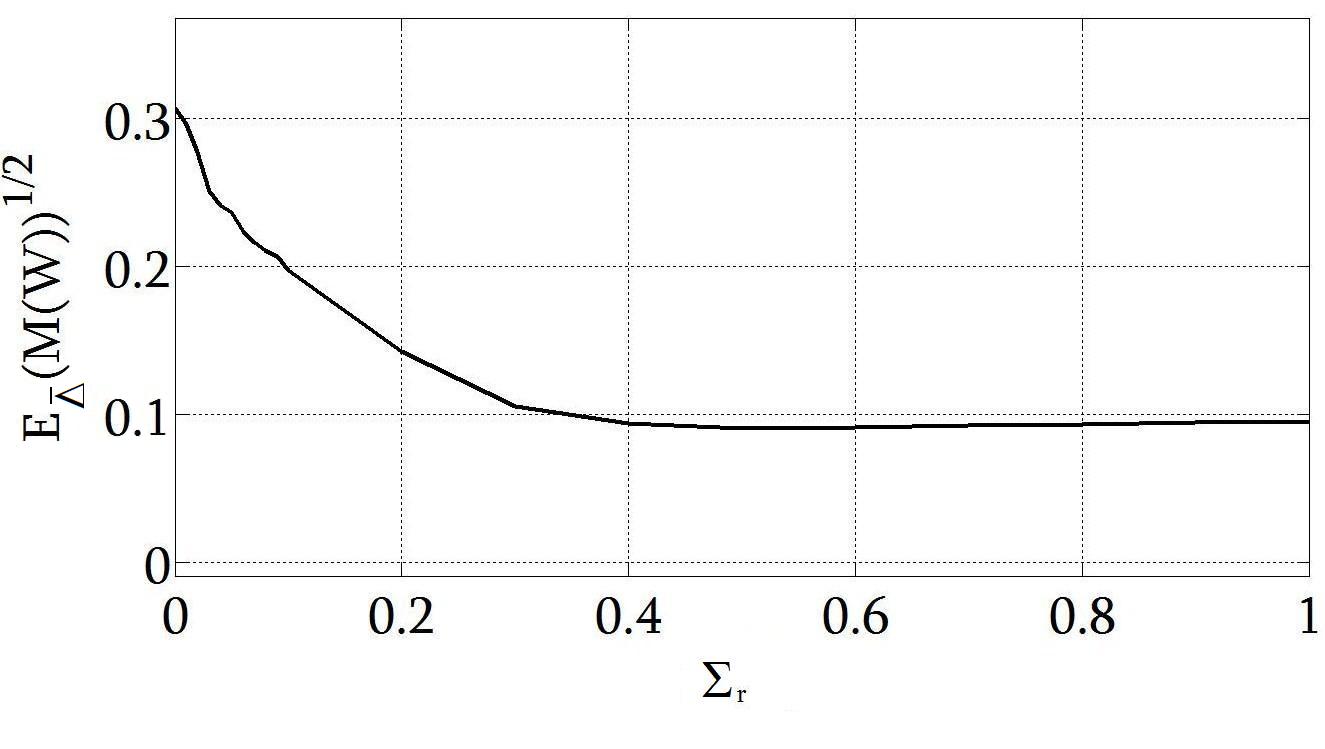}
	\caption{Bayes linear variance learning. $E_{\Dbar}(\M{W_\X})$ as a function of local variance $\musiglc$ for the historical inspection data, shown on square-root scale.}\label{reallocalvsgeneralbounds}
	\end{figure}
	\begin{figure}
	\includegraphics[width=\textwidth]{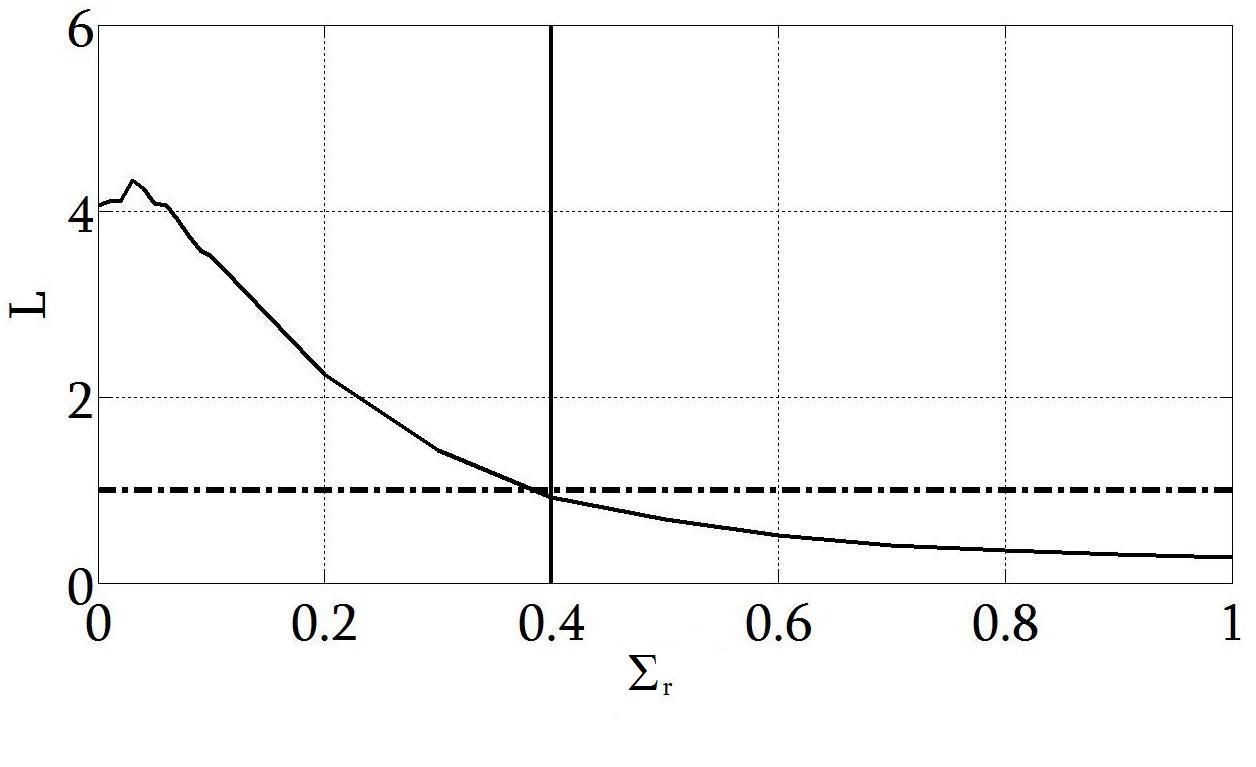}
	\caption{Discrepancy ratio, $H$, as a function of local corrosion variance $\musiglc$ for the historical inspection data. Figure suggests local corrosion variance $\musiglc=0.4^2$ and corresponding adjusted general corrosion variance  $E_{\Dbar}(\M{W_\X})=0.1^2$ (from figure \ref{reallocalvsgeneralbounds}).} \label{realMahalanobisbounds}
	\end{figure}
	
	\begin{figure}
	\includegraphics[width=\textwidth]{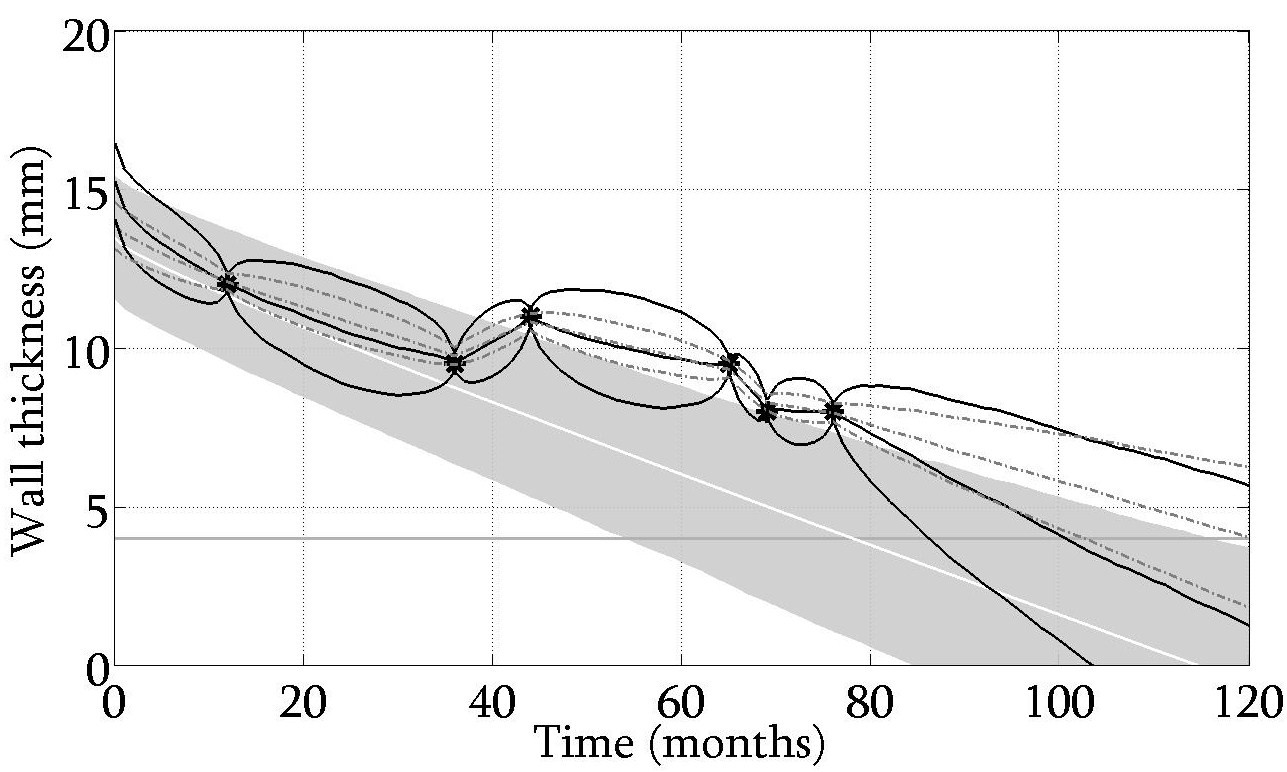}
	\caption{Effect of incorporating variance learning for a single component of the offshore platform. Critical wall thickness is shown as a horizontal dotted line at $4$mm. Actual inspections of the component are shown as black circles. Grey shaded area corresponds to a $95\%$ uncertainty bands from simulation for underlying wall thickness based on prior beliefs (with mean shown as a white line). Light grey broken lines correspond to adjusted expectation $E_Y(Z^{min}_{\comp\ti})$ and $95\%$ uncertainty bands for underlying wall thickness based on Bayes linear adjusted variance $\var_Y(Z^{min}_{\comp\ti})$ without variance learning. Solid black lines correspond to the same adjusted values following a variance learning step.} \label{methodscombined}
	\end{figure}
	The forward model together with parameter estimates from historical inspection data allows forecasting of future system state. Figure \ref{methodscombined} illustrates the effect of incorporating variance learning for a single component of the offshore platform. The critical wall thickness corresponding to component failure is shown as a horizontal dotted line at $4$mm. Actual inspections of the component are shown as black circles. The grey shaded area corresponds to a $95\%$ credible interval for underlying wall thickness based on simulation using prior beliefs (with mean shown as a white line). Light grey broken lines correspond to the adjusted expectation $E_Y(Z^{min}_{ct})$ and $95\%$ credible interval for underlying wall thickness based on Bayes linear adjusted variance $\var_Y(Z^{min}_{ct})$. Solid black lines correspond to the same adjusted values following a variance learning step. Since variance learning inflates local corrosion variance, credible intervals incorporating variance learning are wider. The rate of wall thickness loss, and the expected time to crossing of critical wall thickness are also affected by the incorporation of variance learning. Notwithstanding large uncertainties, in this application, variance learning suggests that component life may be approximately $20$ months shorter than otherwise anticipated.
	
Figure \ref{methodscombined2} illustrates the effect of incorporating variance learning for 3 components of the offshore platform, one of which is not directly observed. The first component is the same component as shown in figure \ref{methodscombined}. The grey shaded area corresponds to a $95\%$ uncertainty interval for underlying wall thickness based on simulation using prior beliefs (with mean shown as a white line); these are the same for each component. Light grey broken lines correspond to the adjusted expectations $E_Y(Z^{min}_{ct})$ with $95\%$ uncertainty intervals based on Bayes linear adjusted variance $\var_Y(Z^{min}_{ct})$. Solid black lines correspond to the same adjusted values following a variance learning step.  The effect of the correlation between components can be seen in changes exhibited by adjusted wall thickness for components at times with no inspections. The third component shows the effect on a component which is not directly observed. In the case of mean updating (without variance learning), the corrosion rate is reduced, increasing predicted remnant life. In the case of variance updating and mean updating, the local corrosion rate is increased and the corrosion rate is reduced. The net effect is to decrease predicted remnant life.

\begin{figure}
	\includegraphics[width=\textwidth]{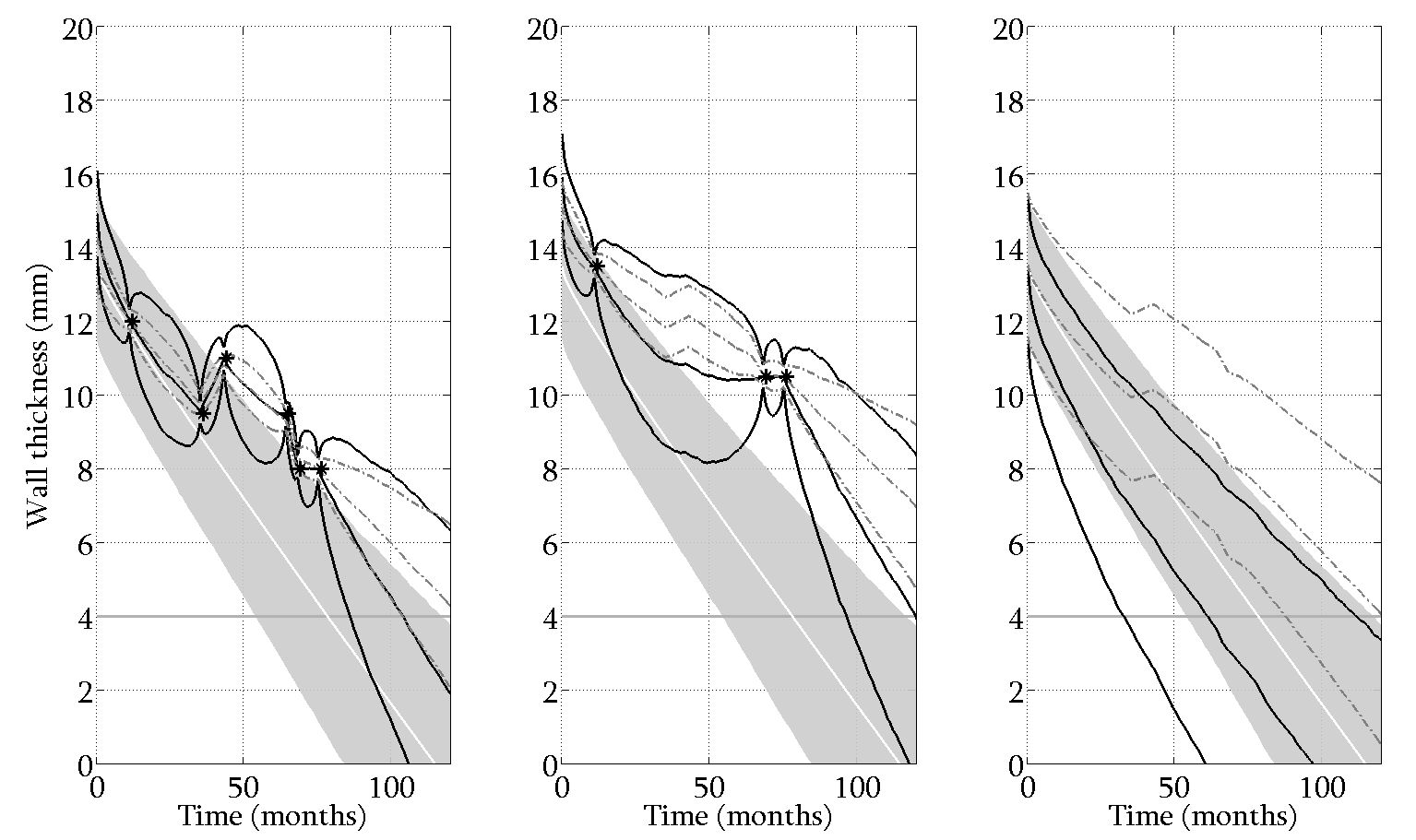}
\caption{Effect of incorporating variance learning for 3 components of the offshore platform, one of which is not directly observed. Critical wall thickness is shown as a horizontal dotted line at $4$mm. Actual inspections of components are shown as black stars. Grey shaded areas corresponds to $95\%$ uncertainty bands from simulation for underlying wall thickness based on prior beliefs (with mean shown as a white line). Light grey broken lines correspond to adjusted expectations $E_Y(Z^{min}_{\comp\ti})$ with $95\%$ uncertainty bands based on Bayes linear adjusted variance $\var_Y(Z^{min}_{\comp\ti})$ without variance learning. Solid black lines correspond to the same adjusted values following a variance learning step.}
\label{methodscombined2}
\end{figure}

\section{Discussion} \label{conc}
	
We present a general framework for modelling the evolution of a large scale physical system consisting of multiple dependent components. Sample data consist of short time series of irregularly-timed partial inspections across components. Realistic modelling requires good specification of priors and initial values, particularly if data are sparse. Simulation and Bayes linear analysis provide knowledge of the full system given partial prior specification, considerably more simply than full Bayesian inference. Given second order exchangeability judgements, variance structure learning for short time series is achieved using linear combinations of observations to access squared residuals where possible. When the observation equation involves a non-linear function a Mahalanobis fitting procedure can be used estimate local effects variance.
	
For the model form given in section \ref{model}, once distributional assumptions are made, it is straight forward to simulate realisations to estimate means, variances and covariances. 

Second order exchangeability judgements used are weaker than independence assumptions typically made in classical inference. Exchangeability across components allows us to exploit large numbers of components for variance structure learning, particularly advantageous when data are sparse. Diagnostic tests assess model fit and validity of exchangeability judgements. Specific knowledge of system characteristics allow a more detailed prior specification, e.g. partitioning the system into sets of exchangeable components, or specifying a parametric form for certain variances in time. With sufficient data, between-component exchangeability judgements may not be necessary. Separate inferences about variances of individual components could be made. We illustrate the method in application to corrosion modelling of a full scale offshore facility. Forecasts of remnant life and its uncertainty incorporating variance learning are shown to be materially different to those obtained without variance learning.
	
Various possibilities for generalisations of the model exist. For example, the corrosion model ignores corrosion initiation, which could be incorporated using an appropriate (e.g point process) representation. For applications other than to corrosion, alternatives to the (separable) minimum function in the observation equation and DLM structure will require tailored linear combinations to access squared residuals required for variance structure learning. Nevertheless, the strategy of Bayes linear variance learning, where possible, supplemented by Mahalanobis learning, will be applicable.
	
\section*{Acknowledgement} \label{ack}
	
This work was supported under an UK Engineering and Physical Sciences Research Council CASE studentship and Knowledge Transfer Network award in conjunction with Shell Projects and Technology. The authors acknowledge discussions with Fred Hoeve, Rakesh Paleja and Seiger Terpstra at Shell.  
	

	\begin{appendix}
	
	\pagebreak
	
	\section{Calculations for $E_D(\M{W_\X})$} \label{App:EDMW}
	
	\subsection{$E[\Dbar]$} \label{Edbar}
	The general corrosion DLM (equation \ref{eq:lingrodlm}) can be rewritten:
	\[
	\left( \begin{array}{c}
	\X \\
	\CR\end{array} \right)_{\comp \ti} = \left(\begin{array}{cc}\In & k \\ \On & \In \\  \end{array}\right)\left( \begin{array}{c}
	\X \\
	\CR\end{array} \right)_{\comp(\ti-k)}+\sum_{i=0}^{k-1} \left(\begin{array}{cc}\In & i \\\On & \In \end{array}\right) \left( \begin{array}{c}
	\epixn+\epialpn \\ \epialpn \end{array} \right)_{\comp(\ti-i)}
	\]
	then
	 \begin{align*}
	 \Y^{(1)}_{\comp\ti^\comp_i} &=\Y_{\comp\ti^{\comp}_i}-\Y_{\comp(\ti^{\comp}_i-k_i^c)} \\
	&=\left(\begin{array}{cc}\In & \On \end{array}\right)\X_{\comp\ti^{\comp}_i}-\left(\begin{array}{cc}\In & \On \end{array}\right)\X_{\comp(\ti^{\comp}_i-k_i^c)}+M _{\comp\ti^\comp_i}^{(1)} \\
	&= \left(\begin{array}{cc}\On & k_i^c \end{array}\right)\X_{\comp(\ti^{\comp}_i-l_i^c)}+\left(\begin{array}{cc}\On & k_i^c \end{array}\right)\E_{\comp(\ti^{\comp}_i-k_i^c,\ti^{\comp}_i-l_i^c)}\\
	&+\left(\begin{array}{cc}\In & \On \end{array}\right)\E_{\comp(\ti^{\comp}_i,\ti^{\comp}_i-k_i^c)}+M _{\comp\ti^\comp_i}^{(1)}
	\end{align*}
	where:
	\[
	\E_{\comp(\ti^{\comp}_i,\ti^{\comp}_i-k_i^c)}=\sum_{i=0}^{k_i^c-1} \left(\begin{array}{cc}
	\In & i \\
	\On & \In \end{array}\right) \left( \begin{array}{c}
	\epixn+\epialpn \\
	\epialpn \end{array} \right)_{\comp(\ti^{\comp}_i-i)}
	\]
	then similarly:
	\begin{align*}
	 \Y^{(2)}_{\comp\ti^\comp_i} &=\left(\begin{array}{cc}\In & \On \end{array}\right)\X_{\comp\ti^{\comp}_i}-\left(\begin{array}{cc}\In & \On \end{array}\right)\X_{\comp(\ti^{\comp}_i-l_i^c)}+M_{\comp\ti^\comp_i}^{(2)} \\
	&= \left(\begin{array}{cc}\On & l_i^c \end{array}\right)\X_{\comp(\ti^{\comp}_i-l_i^c)}+\left(\begin{array}{cc}\In & k_i^c
	\end{array}\right)\E_{\comp(\ti^{\comp}_i-k_i^c,\ti^{\comp}_i-l_i^c)}\\
	&+\left(\begin{array}{cc}\In & \On \end{array}\right)\E_{\comp(\ti^{\comp}_i,\ti^{\comp}_i-k_i^c)}+M_{\comp\ti^\comp_i}^{(2)}
	\end{align*}
	Now to eliminate the effects of the wall thickness term we do $k_i^c\Y^{(2)}_{\comp\ti^\comp_i}-l_i^c\Y^{(1)}_{\comp\ti^\comp_i}$:
	\begin{align*}
	 k_i^c\Y^{(2)}_{\comp\ti^\comp_i}-l_i^c\Y^{(1)}_{\comp\ti^\comp_i} &= k_i^c\bigg[ \left(\begin{array}{cc}\On & l_i^c \end{array}\right)\X_{\comp(\ti^{\comp}_i-l_i^c)}+\left(\begin{array}{cc}\In & k_i^c \end{array}\right)\E_{\comp(\ti^{\comp}_i-k_i^c,\ti^{\comp}_i-l_i^c)}\\ &+\left(\begin{array}{cc}\In & \On \end{array}\right)\E_{\comp(\ti^{\comp}_i,\ti^{\comp}_i-k_i^c)}+M_{\comp\ti^\comp_i}^{(2)} \bigg]\\
	& -l_i^c \bigg[ \left(\begin{array}{cc}\On & k_i^c \end{array}\right)\X_{\comp(\ti^{\comp}_i-l_i^c)}+\left(\begin{array}{cc}\On & k_i^c \end{array}\right)\E_{\comp(\ti^{\comp}_i-k_i^c,\ti^{\comp}_i-l_i^c)} \\&+\left(\begin{array}{cc}\In & \On \end{array}\right)\E_{\comp(\ti^{\comp}_i,\ti^{\comp}_i-k_i^c)}+M_{\comp\ti^\comp_i}^{(1)}\bigg]\\
	&= \left(\begin{array}{cc} (k_i^c-l_i^c) & \On \end{array}\right)\E_{\comp(\ti^{\comp}_i,\ti^{\comp}_i-k_i^c)}+\left(\begin{array}{cc}
	k_i^c &
	k_i^c(k_i^c-l_i^c)\end{array}\right)\E_{\comp(\ti^{\comp}_i-k_i^c,\ti^{\comp}_i-l_i^c)}\\
	&+ k_i^c M_{\comp\ti^\comp_i}^{(2)}- l_i^c M_{\comp\ti^\comp_i}^{(1)}\\
	\end{align*}
	We have that:
	\[
	\E_{\comp(\ti^{\comp}_i,\ti^{\comp}_i-k_i^c)}=\sum_{i=0}^{k_i^c-1} \left(\begin{array}{cc}
	\In & i \\
	\On & \In \end{array}\right) \left( \begin{array}{c}
	\epixn+\epialpn \\
	\epialpn \end{array} \right)_{\comp\ti^{\comp}_i-i}
	\]
	So:
	\begin{eqnarray*}
	\var[\E_{\comp(\ti^{\comp}_i,\ti^{\comp}_i-k_i^c)}] &=& \sum_{i=0}^{k_i^c-1} \var \left[\left(\begin{array}{cc}
	\In & i  \\
	\On & \In \end{array}\right) \left( \begin{array}{c}
	\epixn+\epialpn \\
	\epialpn \end{array}
	\right)_{\comp\ti^{\comp}_i-i}  \right] \\
	&=& \sum_{i=0}^{k_i^c-1} \left(\begin{array}{cc}
	\In & i \\
	\On & \In \end{array}\right) \var\left( \begin{array}{c}
	\epixn+\epialpn \\
	\epialpn \end{array} \right)_{\comp\ti^{\comp}_i-i}\left(\begin{array}{cc}
	\In & \On \\
	i & \In \end{array}\right) \\
	&=& \sum_{i=0}^{k_i^c-1} \left(\begin{array}{cc}
	\In & i \\
	\On & \In \end{array}\right) \left( \begin{array}{cc}
	\musigX+\musigCR & \musigCR \\
	\musigCR & \musigCR \end{array} \right)\left(\begin{array}{cc}
	\In & \On \\
	i & \In \end{array}\right) \\
	&=& \left( \begin{array}{cc}
	k_i^c\musigX+\sum_{i=0}^{k_i^c-1}(1 +i)^2\musigCR & \sum_{i=0}^{k_i^c-1}(1+i) \musigCR \\
	\sum_{i=0}^{k_i^c-1}(1+i)\musigCR & k_i^c\musigCR \end{array}
	\right)\\
	&=&\left(
	\begin{array}{cc}
	k_i^c  +\frac{1}{6}k_i^c(k_i^c+1)(2k_i^c+1) \lambda & \frac{1}{2} k_i^c(k_i^c+1) \lambda \\
	\frac{1}{2} k_i^c(k_i^c+1) \lambda & k_i^c \lambda
	                         \end{array}
	\right)\musigX
	\end{eqnarray*}
	Then:
	\begin{center}
	\small
	\begin{align*}
	E[ (k_i^c\Y^{(2)}_{\comp\ti^\comp_i}-l_i^c\Y^{(1)}_{\comp\ti^\comp_i})^2]&= (k_i^c-1)^2 [ k_i^c + \frac{\lambda k_i^c (k_i^c+1) (2k_i^c+1)}{6}] \musigX \\
	&+ (k_i^c)^2 \left[l_i^c - k_i^c + \frac{\lambda (l_i^c-k_i^c) (l_i^c-k_i^c+1) ( 2(l_i^c-k_i^c) +1)}{6} \right] \musigX \\
	&+ 2(k_i^c)^2 \left[k_i^c - l_i^c + \frac{\lambda (l_i^c-k_i^c)(l_i^c-k_i^c +1)}{2} \right] \musigX \\
	&+ (k_i^c)^2 (k_i^c-l_i^c)^2(l_i^c-k_i^c) \lambda \musigX \\
	&+ (l_i^c)^2 E[(M_{\comp\ti^\comp_i}^{(1)})^2] +(k_i^c)^2 E[(M_{\comp\ti^\comp_i}^{(2)})^2]-2k_i^c l_i^c E[M_{\comp\ti^\comp_i}^{(1)}M_{\comp\ti^\comp_i}^{(2)}]\\
	&= \frac{k_i^c l_i^c (k_i^c -l_i^c) (2 \lambda (k_i^c)^2 -2\lambda l_i^c -6 -\lambda)}{6} \musigX \\
	&+ (l_i^c)^2 E[(M_{\comp\ti^\comp_i}^{(1)})^2] +(k_i^c)^2 E[(M_{\comp\ti^\comp_i}^{(2)})^2]-2k_i^c l_i^c E[M_{\comp\ti^\comp_i}^{(1)}M_{\comp\ti^\comp_i}^{(2)}]\\
	\end{align*}
	\end{center}

	\subsection{$ \cov [\M{W_\X}, \Dbar]$} \label{cov}
	\[\cov [\M{W_\X}, \Dbar]=\left(\cov [\M{W_\X}, \Dbar_1],\cov [\M{W_\X}, \Dbar_2], \dots,\cov [\M{W_\X}, \Dbar_n] \right)\]
	\begin{align*}
	\cov [\M{W_\X}, \Dbar_\comp]	&= \sum_{i=3}^{T_\comp}\frac{1}{K_i} \cov \left[ \M{W_\X},  ( k_i^c Y_{\comp\ti^\comp_i}^{(2)}-l_i^c Y_{\comp\ti^\comp_i}^{(1)})^2  \right]
	\end{align*}
	from \ref{Edbar} we know that:
	\begin{center}
	\small
	\begin{align*}( k_i^c Y_{\comp\ti^\comp_i}^{(2)}-l_i^c Y_{\comp\ti^\comp_i}^{(1)})^2&= \bigg[ \left(\begin{array}{cc} (k_i^c-l_i^c) & \On \end{array}\right)\E_{\comp(\ti^{\comp}_i,\ti^{\comp}_i-k_i^c)}+\left(\begin{array}{cc} k_i^c &
	k_i^c(k_i^c-l_i^c)
	\end{array}\right)\E_{\comp(\ti^{\comp}_i-k_i^c,\ti^{\comp}_i-l_i^c)}\\
	&+ k_i^c M_{\comp\ti^\comp_i}^{(2)}- l_i^c M_{\comp\ti^\comp_i}^{(1)}\bigg]^2
	\end{align*}
	\end{center}
	which means:
	\begin{center}
	\footnotesize
	\begin{align*}
	\cov \left[ \M{W_\X},  ( k_i^c Y_{\comp\ti^\comp_i}^{(2)}-l_i^c Y_{\comp\ti^\comp_i}^{(1)})^2  \right] &=  \cov \left[ \M{W_\X}, \left(  \left(\begin{array}{cc} (k_i^c-l_i^c) & \On \end{array}\right)\E_{\comp(\ti^{\comp}_i,\ti^{\comp}_i-k_i^c)} \right)^2  \right] \\
	& + \cov \left[ \M{W_\X}, \left(  \left(\begin{array}{cc} k_i^c & k_i^c(k_i^c-l_i^c)
	\end{array}\right)\E_{\comp(\ti^{\comp}_i-k_i^c,\ti^{\comp}_i-l_i^c)}\right)^2 \right ]
	\end{align*}
	\end{center}
	we will consider:
	
	A:
	\[\cov \left[ \M{W_\X}, \left(  \left(\begin{array}{cc} (k_i^c-l_i^c) & \On \end{array}\right)\E_{\comp(\ti^{\comp}_i,\ti^{\comp}_i-k_i^c)} \right)^2  \right]\]
	B:
	\[\cov \left[ \M{W_\X}, \left(  \left(\begin{array}{cc} k_i^c &
	k_i^c(k_i^c-l_i^c)
	\end{array}\right)\E_{\comp(\ti^{\comp}_i-k_i^c,\ti^{\comp}_i-l_i^c)}\right)^2 \right ] \]
	%
	A:
	\begin{center}
	\footnotesize
	\begin{align*}
	&\cov \left[ \M{W_\X}, \left(  \left(\begin{array}{cc} (k_i^c-l_i^c) & \On \end{array}\right)\E_{\comp(\ti^{\comp}_i,\ti^{\comp}_i-k_i^c)}
	\right)^2  \right]\\
	&= \cov \left[ \M{W_\X}, \left(  \sum_{j=0}^{k_i^c-1} \left(\begin{array}{cc} (k_i^c-l_i^c) & \On
	\end{array}\right) \left(\begin{array}{cc}
	\In & j \\
	\On & \In \end{array}\right) \left( \begin{array}{c}
	\epixn+\epialpn \\
	\epialpn \end{array} \right)_{\comp\ti^{\comp}_i-j}\right)^2  \right]\\
	&= \cov \left[ \M{W_\X},  \sum_{j=0}^{k_i^c-1}\sum_{j'=0}^{k_i^c-1} \begin{array}{l} (k_i^c-l_i^c)^2\epix{\comp \ti^{\comp}_i-j}\epix{\comp \ti^{\comp}_i-j'}\\+(j'+1)(k_i^c-l_i^c)^2 \epix{\comp \ti^{\comp}_i-j}\epialp{\comp \ti^{\comp}_i-j'}+\\ +(j'+1)(k_i^c-l_i^c)^2 \epix{\comp \ti^{\comp}_i-j}\epialp{\comp \ti^{\comp}_i-j'} +\\+ (j'+1)(j+1)(k_i^c-l_i^c)^2 \epialp{\comp \ti^{\comp}_i-j}\epialp{\comp \ti^{\comp}_i-j'} \end{array} \right] \\
	&=  (k_i^c-l_i^c)^2 \sum_{j=0}^{k_i^c-1} \cov \left[ \M{W_\X},  \epix{\comp \ti^{\comp}_i-j}^2 \right]\\
	&+  (k_i^c-l_i^c)^2 \sum_{j=0}^{k_i^c-1} (j+1)^2 \cov \left[\M{W_\X}, \epialp{\comp \ti^{\comp}_i-j}^2 \right]\\
	&=  (k_i^c-l_i^c)^2 \sum_{j=0}^{k_i^c-1} \cov \left[ \M{W_\X}, \M{W_\CR}+\Res{\comp}{W_\CR}+\Res{\ti}{V_{\CR\comp}}  \right]\\
	&+  (k_i^c-l_i^c)^2 \sum_{j=0}^{k_i^c-1} (j+1)^2 \cov \left[\M{W_\X},  \M{W_\CR}) \right]\\
	&= k_i^c(k_i^c-l_i^c)^2 \covsigX+ \frac{\lambda}{6} k_i^c(k_i^c+1)(2k_i^c+1)(k_i^c-l_i^c)^2 \covsigX
	\end{align*}
	\end{center}
	%
	B:
	\begin{center}
	\footnotesize
	\begin{align*}
	&\cov \left[ \M{W_\X}, \left(  \left(\begin{array}{cc} k_i^c & k_i^c(k_i^c-l_i^c)
	\end{array}\right)\E_{\comp(\ti^{\comp}_i-k_i^c,\ti^{\comp}_i-l_i^c)}\right)^2 \right ]\\
	&= \cov \left[ \M{W_\X}, \left(\sum_{j=0}^{l_i^c-k_i^c-1}  \left(\begin{array}{cc} k_i^c & k_i^c(k_i^c-l_i^c)
	\end{array}\right)\left(\begin{array}{cc}
	\In & j \\
	\On & \In \end{array}\right) \left( \begin{array}{c}
	\epixn+\epialpn \\
	\epialpn \end{array} \right)_{\comp\ti^{\comp}_i-k_i^c-j}\right)^2 \right ]\\
	&= \cov \left[ \M{W_\X}, \sum_{j=0}^{l_i^c-k_i^c-1} \sum_{j'=0}^{l_i^c-k_i^c-1}  \begin{array}{l}\left(k_i^c
	\epix{\comp\ti^{\comp}_i-k_i^c-j} + k_i^c(j+1+k_i^c-l_i^c) \epialp{\comp\ti^{\comp}_i-k_i^c-j}\right)\\ \times \left(k_i^c
	\epix{\comp\ti^{\comp}_i-k_i^c-j'} + k_i^c(j'+1+k_i^c-l_i^c) \epialp{\comp\ti^{\comp}_i-k_i^c-j'}\right) \end{array}
	\right ]\\
	&= (l_i^c-k_i^c) (k_i^c)^2 \covsigX + \cov \left[ \M{W_\X}, (k_i^c)^2 \sum_{j=0}^{l_i^c-k_i^c-1} \begin{array}{l}j^2 \epialp{\comp\ti^{\comp}_i-k_i^c-j}^2\\ +2j(k_i^c-l_i^c) \epialp{\comp\ti^{\comp}_i-k_i^c-j}^2 \\ +(k_i^c-l_i^c) \epialp{\comp\ti^{\comp}_i-k_i^c-j}^2\end{array}\right ]\\
	&= (l_i^c-k_i^c) (k_i^c)^2 \covsigX + \frac{\lambda (k_i^c)^2 (l_i^c-k_i^c)(l_i^c-k_i^c+1)(2(l_i^c-k_i^c)+1) \covsigX }{6}\\
	& + \frac{2\lambda (k_i^c)^2(l_i^c-k_i^c)(l_i^c-k_i^c+1)(k_i^c-l_i^c)\covsigX}{2} +\lambda (k_i^c)^2 (k_i^c-l_i^c)^2(l_i^c-k_i^c) \covsigX
	\end{align*}
	\end{center}
	So:
	\begin{eqnarray*}
	 \cov \left[ \M{W_\X},  ( k_i^c Y_{\comp\ti^\comp_i}^{(2)}-l_i^c Y_{\comp\ti^\comp_i}^{(1)})^2  \right]&=& \textrm{A.}+\textrm{B.}\\
	&=& \frac{k_i^c l_i^c (k_i^c -l_i^c) (2 \lambda (k_i^c)^2 -2\lambda l_i^c -6 -\lambda)\covsigX}{6}\\
	&=& K_i \covsigX
	\end{eqnarray*}
	Therefore:
	\begin{eqnarray*}
	\cov [\M{W_\X}, \Dbar_\comp]&=&  \sum_{i=3}^{T_c} \frac{K_i}{K_i} \covsigX \\
	&=& (T_c-2)\covsigX
	\end{eqnarray*}
	and:
	\[ \cov [\M{W_\X}, \Dbar ] = \left(\begin{array}{cccc} \covsigX & \covsigX &\dots & \covsigX \end{array}\right)\]

	\section{Prior values for offshore structure application}
	 \label{priorreal}
	
	\begin{center}
	\small
	\begin{tabular}{lll}
	number of components & $N$ & $64$ \\
	number of time points & $T$ &$ 83$ \\
	total number of inspections &  & $174$ \\
	minimum level of correlation between components & $\rho_0$ & $0.2$\\
	distance effect scaling parameter& $\rho_D$ & $0.3$\\
  circuit correlation scaling parameter& $\rho_C$ & $0.5$\\
	ratio local to general corrosion &$\lambda$& $0.02$ \\
	measurement error variance & $\Sigy $ &$0.16^2$\\
	prior mean for $\M{W_\X}$ & $\musigX$ &$0.1^2$\\
	prior variance for $\M{W_\X}$ & $ \Sigma_{W_\X} $ & 1e-3\\
	prior covariance for $\M{W_\X}$ & $ \Gamma_{W_\X} $ & 5e-4
	\end{tabular}
	\end{center}
	\end{appendix}
	\end{document}